\documentclass[letterpaper,11pt]{article} 
\usepackage{fullpage}
\usepackage{amsthm}
\usepackage{amsmath}
\usepackage{amssymb}
\usepackage{latexsym}
\usepackage{epsfig}
\usepackage{cite}
\usepackage{color}
\usepackage{subfig}
\usepackage{graphicx}
\usepackage{paralist}
\usepackage[colorlinks,urlcolor=blue,citecolor=blue,linkcolor=blue]{hyperref}

\newtheorem{theorem}{Theorem}[section]
\newtheorem{lemma}[theorem]{Lemma}
\newtheorem{claim}[theorem]{Claim}

\newtheorem{definition}[theorem]{Definition}
\newtheorem{prop}[theorem]{Proposition}
\newtheorem{corol}[theorem]{Corollary}

\mathchardef\lee="0214

\newcommand{\wlambda}{\widehat{\lambda}}
\newcommand{\eps}{\varepsilon}
\newcommand{\etal}{et al.}
\newcommand{\eqdef}{:=}
\newcommand{\EX}{\hbox{\bf E}}
\newcommand{\R}{\mathbb{R}}
\newcommand{\OPTMAX}{\text{OPT-MAX}}
\newcommand{\OPTCH}{\text{OPT-CH}}

\DeclareMathOperator{\UH}{conv}
\DeclareMathOperator{\seg}{seg}
\DeclareMathOperator{\vis}{vis}
\DeclareMathOperator{\poly}{poly}
\DeclareMathOperator{\pen}{pen}
\DeclareMathOperator{\reg}{reg}
\DeclareMathOperator{\lev}{lev}
\DeclareMathOperator{\uss}{uss}
\DeclareMathOperator{\lss}{lss}

\newcommand{\findmax}{\texttt{find-max}}
\newcommand{\ins}{\texttt{insert}}
\newcommand{\delete}{\texttt{delete}}
\newcommand{\deckey}{\texttt{decrease-key}}

\newcommand{\bS}{\textbf{S}}
\newcommand{\bT}{\textbf{T}}
\newcommand{\bU}{\textbf{U}}

\newcommand{\cC}{\mathcal{C}}
\newcommand{\cD}{\mathcal{D}}
\newcommand{\cR}{\mathcal{R}}
\newcommand{\cE}{\mathcal{E}}

\newcommand{\cF}{\mathcal{F}}

\newcommand{\cT}{\mathcal{T}}

\title{
  Self-improving Algorithms for Coordinate-Wise Maxima and Convex 
  Hulls\footnote{Preliminary versions appeared as
  K. L. Clarkson, W. Mulzer, and C. Seshadhri,
  \emph{Self-improving Algorithms for Convex Hulls} in Proc.~21st SODA, 
  pp.~1546--1565, 2010;
  and K. L. Clarkson, W. Mulzer and C. Seshadhri, 
  \emph{Self-improving Algorithms for
  Coordinate-wise Maxima} in Proc.~28th SoCG, pp.~277--286, 2012.
  }
}

\author{
  Kenneth L. Clarkson\footnote{IBM Almaden Research Center, San Jose, USA.
    \texttt{klclarks@us.ibm.com}
  }
  \and
  Wolfgang Mulzer\footnote{Institut f\"ur Informatik, Freie Universit\"at 
    Berlin, Berlin, Germany.
    \texttt{mulzer@inf.fu-berlin.de} 
  }
  \and
  C. Seshadhri\footnote{Sandia National Laboratories, Livermore, USA.
    \texttt{scomand@sandia.gov}
}
}

\begin{document} 

\maketitle

\begin{abstract}
  Finding the coordinate-wise maxima and the convex 
  hull of a planar point set are probably the most 
  classic problems in computational geometry.  We 
  consider these problems in the \emph{self-improving setting}. 
  Here, we have $n$ distributions $\cD_1, \ldots, \cD_n$ 
  of planar points. An input point set $(p_1, \ldots, p_n)$ 
  is generated by taking an independent sample $p_i$ 
  from each $\cD_i$, so the input is distributed 
  according to the product $\cD = \prod_i \cD_i$. 
  A \emph{self-improving algorithm} repeatedly gets inputs 
  from the distribution $\cD$ (which is \emph{a priori} 
  unknown), and it tries to optimize its running time 
  for $\cD$. The algorithm uses the first few inputs to 
  learn salient features of the distribution $\cD$, before 
  it becomes fine-tuned to $\cD$. Let $\OPTMAX_\cD$ 
  (resp.~$\OPTCH_\cD$) be the expected depth 
  of an \emph{optimal} linear comparison tree computing 
  the maxima (resp.~convex hull) for $\cD$. 
  Our maxima algorithm eventually achieves expected 
  running time $O(\OPTMAX_\cD + n)$. 
  Furthermore, we give a self-improving algorithm for convex hulls
  with expected running time $O(\OPTCH_\cD + n\log\log n)$.

  Our results require new tools for understanding linear 
  comparison trees. In particular, we convert 
  a general linear comparison tree to a restricted 
  version that can then be related to the running time 
  of our algorithms. Another interesting feature 
  is an interleaved search procedure 
  to determine the likeliest point to be extremal 
  with minimal computation. This allows our algorithms
  to be competitive with the optimal algorithm for 
  $\cD$.
\end{abstract}

\section{Introduction}
\label{sec:intro}

The problems of planar maxima and planar convex hull 
computation are classic computational geometry questions 
that have been studied since at least 1975~\cite{KungLuPr75}. 
There are well-known $O(n \log n)$ time comparison-based 
algorithms ($n$ is the number of input points), with 
matching lower bounds. Since then,
many more advanced settings have been addressed: one can 
get expected running time $O(n)$ for points uniformly 
distributed in the unit square; 
output-sensitive algorithms need $O(n\log h)$ time
for output size $h$~\cite{KirkpatrickSe86}; 
and there are results for external-memory 
models~\cite{GoodrichTsVeVi93}. 

A major drawback of worst-case analysis
is that it does not always reflect the behavior
of real-world inputs. Worst-case algorithms 
must provide for extreme inputs that 
may not occur (reasonably often) in 
practice. Average-case analysis tries 
to address this problem by assuming some 
fixed input distribution. For example, in the 
case of maxima 
coordinate-wise independence covers a broad
range of inputs, and it leads to a clean 
analysis~\cite{Buchta89}. Nonetheless, it is still
unrealistic, and the right 
distribution to analyze remains a point of 
investigation. However, the assumption of 
randomly distributed inputs is very natural 
and one worthy of further study.

\paragraph{The self-improving model.} 
Ailon~\etal~introduced the self-improving 
model to address the drawbacks of average 
case analysis~\cite{ACCL}. 
In this model, there is a fixed but 
unknown distribution $\cD$ that 
generates independent inputs, i.e., whole 
input sets $P$. The algorithm initially
undergoes a \emph{learning phase} where 
it processes inputs with a worst-case
guarantee while acquiring information 
about $\cD$. 
After seeing a (hopefully small) 
number of inputs, the algorithm shifts
into the \emph{limiting phase}. Now, it 
is tuned for $\cD$, and the expected 
running time is (ideally) \emph{optimal 
for the distribution $\cD$}. 
A self-improving algorithm 
can be thought of as able to attain 
the optimal average-case
running time for all, or at least 
a large class of, distributions. 

As in earlier work, 
we assume that the input follows a product 
distribution. An input $P = (p_1, \ldots, p_n)$
is a set of $n$ points in the plane. Each 
$p_i$ is generated independently from a 
distribution $\cD_i$, so the probability 
distribution of $P$ is the product 
$\prod_i \cD_i$. The $\cD_i$s themselves
are arbitrary, we only assume that they 
are independent. There are lower 
bounds~\cite{AilonCCLMS11} showing that 
some restriction on $\cD$ is
necessary for a reasonable self-improving 
algorithm, as we shall explain below.

The first self-improving algorithm was 
for sorting, and it was later extended 
to Delaunay 
triangulations~\cite{CS_self_improve,AilonCCLMS11}.
In both cases, \emph{entropy-optimal} 
performance is achieved in the limiting phase.
Later, Bose~\etal~\cite{BoseDeDoDuKiMo10} 
described \emph{odds-on trees}, a general 
method for self-improving solutions to
certain query problems, e.g., 
point location, orthogonal range searching, 
or point-in-polytope queries.

\section{Results} 
\label{sec:results}

We give self-improving algorithms 
for planar coordinate-wise maxima
and convex hulls over product 
distributions. 
Let $P \subseteq \R^2$ be finite.
A point $p \in P$
\emph{dominates} $q \in P$,
if both the $x$- and $y$-coordinate
of $p$ are at least as large as
the $x$- and $y$-coordinate of $q$.
A point in $P$ is \emph{maximal}
if no other point in $P$ dominates
it, and \emph{non-maximal} otherwise.
The \emph{maxima problem} is
to find all maximal points in $P$.
The \emph{convex hull} of $P$ is the
smallest convex set that contains $P$.
It is a convex polygon whose vertices
are points from $P$. 
We will focus on the \emph{upper} hull of $P$,
denoted by $\UH(P)$.  A point in $P$
is \emph{extremal} if it appears on 
$\UH(P)$, otherwise it is 
\emph{non-extremal}. In the \emph{convex
hull problem}, we must find the extremal
points in $P$.

\subsection{Certificates}

We need to make precise the notion of an \emph{optimal algorithm}
for a distribution $\cD$. The 
issue with maxima and convex hulls
is their output sensitive nature.
Even though the actual output size
may be small, additional work 
is necessary to determine which
points appear in the output. We also want to consider algorithms
that give a correct output on \emph{all} instances, not just those
in the support of $\cD$. For example, suppose for all inputs in the support of $\cD$,
there was a set of (say) three points that always formed the maxima. The optimal
algorithm just for $\cD$ could always output these three points. But such an algorithm
is not a legitimate maxima algorithm, since it would be incorrect on other inputs.

To handle these issues, we demand that any algorithm must provide a simple proof
that the output is correct.
This is formalized through \emph{certificates} (see Fig.~\ref{fig:certs}).

\begin{figure}
  \centering
  \includegraphics{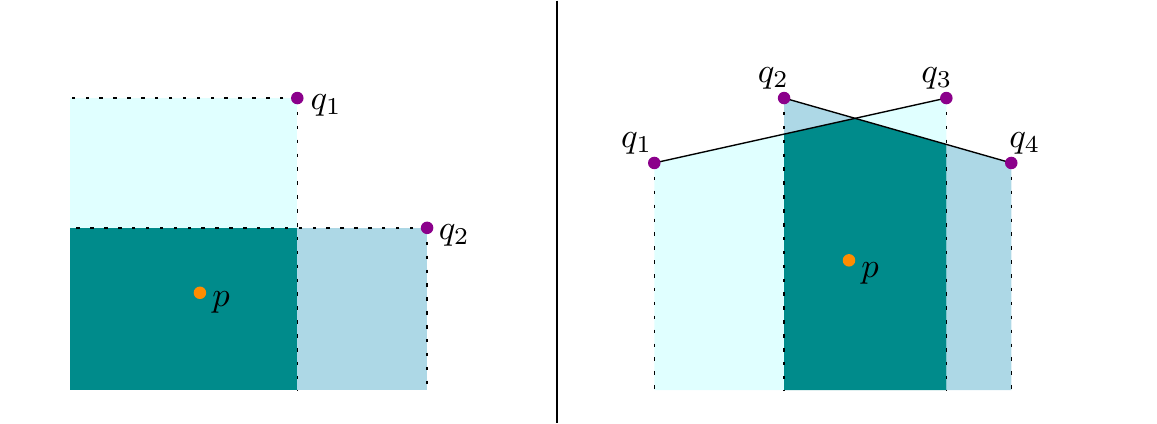}

   \caption{Certificates for maxima 
     and convex hulls: (left) both 
     $q_1$ and $q_2$ are certificates of 
     non-maximality for $p$; (right) 
     both $q_1q_3$ and $q_2q_4$ 
     are possible witness pairs 
     for non-extremality of $p$.
   }
  \label{fig:certs}
\end{figure}

\begin{definition}\label{def:cert-max}
  Let $P \subseteq \R^2$  be finite.
  A \emph{maxima certificate} $\gamma$ 
  for $P$ consists of \textup(i\textup) 
  the indices of the maximal points 
  in $P$, sorted from left to right; 
  and \textup(ii\textup) a \emph{per-point 
  certificate} for each non-maximal 
  point $p \in P$, i.e., the index 
  of an input point that dominates 
  $p$.  A certificate $\gamma$ is 
  \emph{valid} for $P$ if $\gamma$ 
  satisfies conditions (i) and (ii) 
  for $P$.
\end{definition}

Most known algorithms implicitly 
provide a certificate as in 
Definition~\ref{def:cert-max}~\cite{KungLuPr75,Golin94,KirkpatrickSe86}. 
For two points $p, q \in P$, we 
define the \emph{upper semislab} 
for $p$ and $q$, $\uss(p,q)$,
as the open planar region 
bounded by the upward vertical 
rays through $p$ and $q$ and 
the line segment $\overline{pq}$. 
The \emph{lower semislab} for 
$p$ and $q$, $\lss(p,q)$,
is defined analogously. Two 
points $q,r \in P$ are a 
\emph{witness pair} for a 
non-extremal $p \in P$ if 
$p \in \lss(q,r)$.

\begin{definition}\label{def:cert-ch}
  Let $P \subseteq \R^2$ be finite. A 
  \emph{convex hull certificate} $\gamma$ 
  for $P$ has \textup(i\textup) 
  the extremal points in $P$, sorted from 
  left to right; and \textup(ii\textup) a 
  witness pair for each non-extremal point 
  in $P$. The points in $\gamma$ are 
  represented by their indices in $P$.
\end{definition}

To our knowledge, most current maxima and convex hull algorithms implicitly output 
such certificates (for example, when they prune non-extremal points). This is by no means
the only possible set of certificates, and one could design different types of certificates. Our notion
of optimality crucially depends on the definition of certificates. It is not \emph{a priori} clear
how to define optimality with respect to other definitions, though we feel that our certificates are quite
natural.

\subsection{Linear comparison trees}

To define optimality, we need 
a lower bound model to which
our algorithms can be compared.
For this, we use linear algebraic 
computation trees that perform 
comparisons according to query lines 
defined by the input points. 
Let $\ell$ be a directed line. We write 
$\ell^+$ for the open halfplane
to the left of $\ell$, and $\ell^-$ for 
the open halfplane to the right of $\ell$. 

\begin{definition}\label{def:opt} 
  A \emph{linear comparison tree} $\cT$ 
  is a rooted binary tree. Each 
  node $v$ of $\cT$ is labeled with 
  a query of the form ``$p \in \ell_v^+?$''. 
  Here, $p$ is an input point and 
  $\ell_v$ a directed line. The 
  line $\ell_v$ can be obtained 
  in four ways, in increasing complexity:
  \begin{asparaenum}
    \item\label{type:I} 
      a fixed line independent of 
      the input \textup(but dependent
      on $v$\textup);
     \item\label{type:II} 
       a line with a fixed slope 
       \textup(dependent on $v$\textup) 
       passing through a given input point;
     \item\label{type:III} 
       a line through an input point and 
       a fixed point $q_v$, dependent on $v$; or
     \item\label{type:IV} 
       a line through two distinct input points. 
  \end{asparaenum}
\end{definition}

\begin{figure}
  \centering
  \includegraphics{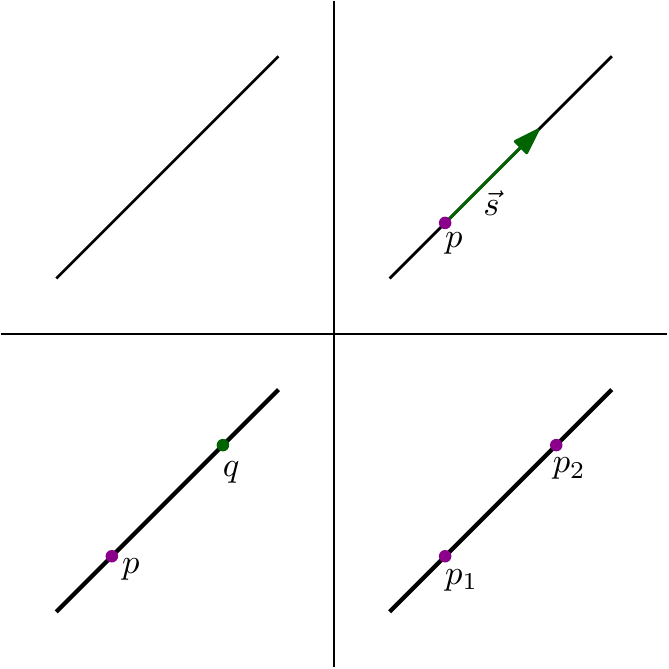}
  \caption{
    We can compare with (a) a fixed line;
    (b) a line through input point $p$, 
    with fixed slope $\protect\overrightarrow{s}$; 
    (c) a line through input $p$ and a 
    fixed point $q$; and (d) a line through 
    inputs $p_1$ and $p_2$.
  }
  \label{fig:lincomp}
\end{figure}

Definition~\ref{def:opt} is illustrated 
in Fig.~\ref{fig:lincomp}. 
Given an input $P$, an 
\emph{evaluation} of a linear 
comparison tree $\cT$ on $P$ 
is the node sequence 
that starts at the root 
and chooses in each step
the child according to the 
outcome of the current 
comparison on $P$. 
For a node $v$ of $\cT$
there is a region 
$\cR_v \subseteq \R^{2n}$ such 
that an evaluation of $\cT$ 
on input $P$ reaches $v$ if and only 
if $P \in \cR_v$. 

Why do we choose this model? 
For starters, it captures the standard ``counter-clockwise" (CCW) primitive. 
This is the is-left-of test that checks
whether a point $p$ lies to the left, 
on, or to the right of the directed line 
$qr$, where $p$, $q$, and $r$ are
input points~\cite{deBergChvKrOv08}.
The model also contains simple coordinate comparisons, the usual operation for maxima finding. Indeed, most planar maxima and convex hull algorithms only use these operations. Since we are talking about distributions of points, it also makes sense (in our opinion) to consider comparisons with fixed lines. All our definitions of optimality are dependent on this model,
so it would be interesting to extend our results to more general models. We may consider
comparisons with lines that have more complex dependences on the input points. Or, consider
relationships with more than 3 points. Nonetheless, this model is a reasonable starting point
for defining optimal maxima and convex hull algorithms.

We can now formalize linear
comparison trees for maxima and
convex hulls.

\begin{definition}
  A  linear comparison tree $\cT$ 
  \emph{computes the maxima
  of a planar point set} if
  every leaf $v$ of $\cT$ is labeled
  with a maxima certificate that is 
  valid for every input $P \in \cR_v$.
  A linear comparison tree for planar
  convex hulls is defined analogously.
\end{definition}

The \emph{depth} $d_v$ of node $v$ in $\cT$ 
is the length of the path from the root 
of $\cT$ to $v$.  Let $v(P)$  be the 
leaf reached by the evaluation of 
$\cT$ on input $P$.  The 
\emph{expected depth} of $\cT$ over 
$\cD$ is defined as 
\[
  d_\cD(\cT) = \EX_{P \sim \cD}[d_{v(P)}]. 
\]
For a comparison based algorithm whose
decision structure is modeled by $\cT$, 
the expected depth of $\cT$ gives a lower
bound on the expected running time.

\subsection{Main theorems}

Let $\bT$ be the set of linear comparison
trees that compute the maxima of $n$ points.
We define $\OPTMAX_\cD = \inf_{\cT \in \bT} d_\cD(\cT)$.
$\OPTMAX$ is a lower bound on the 
expected running time of \emph{any} linear
comparison tree to compute the maxima 
according to $\cD$. We prove the following
result:

\begin{theorem}\label{thm:main-max} 
  Let $\eps > 0$ be a fixed constant and 
  $\cD_1, \dots, \cD_n$ continuous planar point 
  distributions. Set $\cD = \prod_i \cD_i$.
  There is a self-improving algorithm for 
  coordinate-wise maxima according to $\cD$
  whose expected time in the limiting phase is 
  $O(\eps^{-1}(n + \OPTMAX_\cD))$. 
  The learning phase takes $O(n^\eps)$ inputs. 
  The space requirement is $O(n^{1+\eps})$.
\end{theorem}

We also give a self-improving algorithm 
for convex hulls. Unfortunately, it is 
slightly suboptimal. Like before, we set 
$\OPTCH_\cD = \inf_{\cT \in \bT} d_\cD(\cT)$, 
where now $\bT$ is the set of linear comparison 
trees for the convex hull of $n$ points. 
The conference version~\cite{ClarksonMuSe10}
claimed an optimal result, but the analysis was incorrect. Our new analysis is simpler and 
closer in style to the maxima result.

\begin{theorem}\label{thm:main-ch} 
  Let $\eps > 0$ be a fixed constant 
  and $\cD_1, \ldots, \cD_n$ 
  continuous planar point distributions. 
  Set $\cD = \prod_i \cD_i$. There is a 
  self-improving algorithm for convex hulls
  according to $\cD$ whose expected time in 
  the limiting phase is 
  $O(n\log\log n + \eps^{-1}(n + \OPTCH_\cD))$. 
  The learning phase takes $O(n^\eps)$ inputs. 
  The space requirement is $O(n^{1+\eps})$.
\end{theorem}

Any optimal (up to multiplicative factor $1/\eps$ in running time)
self-improving sorter requires $n^{1+\Omega(\eps)}$ storage (Theorem 2 of~\cite{AilonCCLMS11}).
By the standard reduction of sorting to maxima and convex hulls, this shows that 
the $O(n^{1+\eps})$ space is necessary.
Furthermore, self-improving sorters for arbitrary distributions requires exponential storage (Theorem 2 of~\cite{AilonCCLMS11}). So some 
restriction on the input distribution is necessary for a non-trivial result.

\begin{figure}
  \centering
  \includegraphics{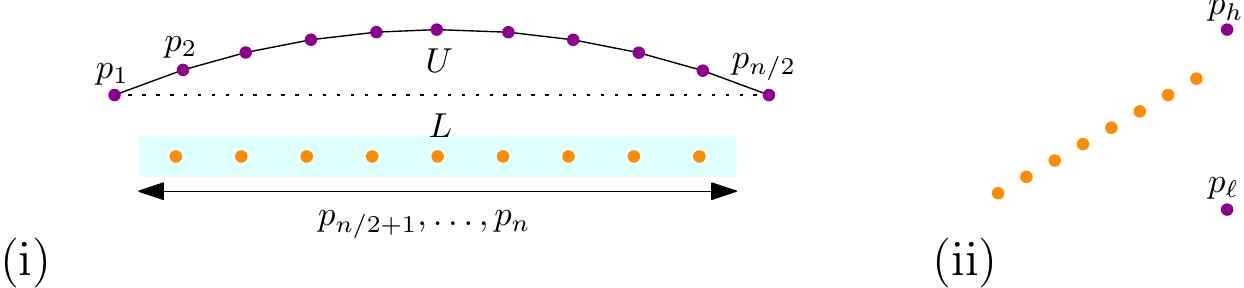}

  \caption{Bad inputs: (i) the 
    upper hull $U$ is fixed, while 
    $p_{n/2 +1}, \ldots, p_n$ roughly 
    constitute a random permutation of $L$; 
    (ii) point $p_1$ is either at $p_h$ 
    or $p_\ell$, so it affects the 
    extremality of the other inputs.}
  \label{fig:bad}
\end{figure}

\paragraph{Prior Algorithms.}
Before we go into the details of
our algorithms, let us explain why
several previous approaches fail.
We focus 
on convex hulls, but the arguments
are equally valid for maxima.
The main problem seems to be
that the previous approaches rely
on the sorting lower bound for 
optimality. However, this 
lower bound does not apply 
in our model.  Refer to 
Fig.~\ref{fig:bad}(i). 
The input comes in two groups: 
the lower group $L$ is not on 
the upper hull, while all points
in the upper group $U$
are vertices of the upper hull. 
Both $L$ and $U$ have
$n/2$ points. The input distribution 
$\cD$ fixes the points 
$p_1, \dots, p_{n/2}$ to form $U$, 
and for each $p_i$ with $i = n/2+1, \dots, n$, 
it picks a random point from $L$ 
(some points of $L$ may be repeated). 
The ``lower'' points form a random 
permutation of $\Omega(n)$ points from $L$. 
The upper hull is always given by $U$, 
while all lower points have the same
witness pair $p_1, p_{n/2}$. Thus 
an optimal algorithm requires $O(n)$ time.

In several other models, the 
example needs $\Omega(n\log n)$ 
time. The output size is $n/2$, so 
output-sensitive algorithms require 
$\Omega(n\log n)$ steps. Also, 
the \emph{structural entropy}
is $\Omega(n\log n)$~\cite{Barbay}. 
Since the expected size of the
upper hull of a random $r$-subset 
of $U\cup L$ is $r/2$, 
randomized incremental construction 
takes $\Theta(n\log n)$ time~\cite{CS89}. 
As the entropy of the 
$x$-ordering is $\Omega(n\log n)$, 
self-improving algorithms
for sorting or Delaunay
triangulations are not
helpful~\cite{AilonCCLMS11}.
\emph{Instance optimal algorithms}
also require $\Omega(n \log n)$
steps for each input from our 
example~\cite{AfshaniBC09}: 
this setting 
considers the input as a \emph{set}, 
whereas for us
it is essential to know the distribution of
each individual input point.

Finally, we mention the paradigm of 
preprocessing imprecise
points~\cite{BuchinLoMoMu11,EzraMu13,HeldMi08,LoefflerSn10,vKreveldLoMi10}.
Given a set $\cR$ of 
planar regions, we must 
preprocess $\cR$ to 
quickly find the (Delaunay) 
triangulation or convex hull 
for inputs with exactly 
one point from each region 
in $\cR$.  If we consider inputs 
with a random point from each 
region, the self-improving setting applies,
and the previous results 
bound the expected
running time in the limiting phase.
As a noteworthy side effect, we 
improve a result
by Ezra and Mulzer~\cite{EzraMu13}:
they preprocess a set of
planar lines so that
the convex hull for 
inputs with one point from each
line can be found in near-linear
time. Unfortunately, the data structure
needs quadratic space.
Using self-improvement,
this can now be reduced
to $O(n^{1+\eps})$.

\paragraph{Output sensitivity and dependencies.}
We introduced certificates in 
order to deal with output 
sensitivity. These certificates
may or may not be easy 
to find. In Fig.~\ref{fig:bad}(i), 
the witness pairs are all 
``easy''. However, if 
the points of $L$ are placed
\emph{just below} the edges 
of the upper hull, we need 
to search for the witness 
pair of each point $p_i$, 
for $i=n/2+1, \dots, n$;
the certificates are
``hard''. Furthermore, 
even though the individual 
points are independent, the 
upper hull can exhibit 
very dependent behavior. 
In Fig.~\ref{fig:bad}(ii),
point $p_1$ can be either 
$p_h$ or $p_\ell$, while 
the other points are fixed. 
The points $p_2, \dots, p_n$ 
become extremal \emph{depending} 
on the position of $p_1$. This 
makes life rather hard for 
entropy-optimality, since only 
if $p_1=p_\ell$ the ordering of 
$p_2,\dots, p_n$ must be determined.

Our algorithm, and plausibly 
any algorithm, performs a 
point location for each input $p_i$.
If $p_i$ is ``easily'' shown to be 
non-extremal, the search 
should stop early. 
However, it seems
impossible to know \emph{a priori}
how far to proceed: imagine 
the points $L$ of Fig.~\ref{fig:bad}(i) 
doubled up and placed at \emph{both} 
the ``hard'' and the ``easy'' positions, 
and $p_i$ for $i=n/2+1, \dots, n$ chosen 
randomly among them. The 
search depth can only be 
determined from the actual 
position. Moreover, the certificates may 
be easy once the extremal points are known,
but finding them is what we wanted in the
first place.

\section{Preliminaries}
\label{sec:prelims}

 Our input point set 
is called 
$P = \langle p_1, \dots, p_n\rangle$, 
and it comes from a product
distribution 
$\cD = \prod_{i=1}^n \cD_i$. 
All distributions $\cD_i$
are assumed to be continuous.
For $p \in \R^2$, we 
write $x(p)$ and $y(p)$ for the 
$x$- and the $y$-coordinate 
of $p$.  Recall that 
$\ell^+$ and $\ell^-$ denote
the open halfplanes to
the left and to the right of
a directed line $\ell$.
If $R \subseteq \R^2$ is measurable, 
a \emph{halving line} $\ell$
for $R$ with respect to distribution
$\cD_i$ has the property 
\[
  \Pr_{p \sim \cD_i}[p \in \ell^+ \cap R] = 
  \Pr_{p \sim \cD_i}[p \in \ell^- \cap R].
\]
If $\Pr_{p \sim \cD_i}[p \in R] = 0$, every 
line is a halving line for $R$. 

We write $c$ for a sufficiently 
large constant.
We say ``with high probability" for 
any probability larger than 
$1 - n^{-\Omega(1)}$.
The constant in the exponent 
can be increased by increasing
the constant $c$. We will take union 
bounds over polynomially many 
(usually at most $n^2$) 
low probability events and still get a 
low probability bound.

The main self-improving algorithms require a significant amount of preperation. This is detailed in Sections \ref{sec:comp}, \ref{sec:rest}, and \ref{sec:data-str}. These sections give some lemmas
on the linear comparison trees, search trees, and useful data structures for the learning phase.
We would recommend the reader to first
skip all the proofs in these sections, as they are somewhat unrelated to the actual self-improving algorithms.

\section{Linear comparison trees} \label{sec:comp}

We discuss  basic properties 
of linear comparison trees.
Crucially,  
any such tree can be
simplified without significant loss in
efficiency (Lemma~\ref{lem:lin->entropy}).
Let $\cT$ be a linear comparison tree.
Recall that for each node $v$ of $\cT$, 
there is an \emph{open} region 
$\cR_v \subseteq \R^{2n}$ such that 
an evaluation of $\cT$ on $P$ reaches $v$ 
if and only if $P \in \cR_v$
(We define the regions as open,
because the continuous nature of
the input distribution lets us 
ignore the case that a point
lies on a query line.)
We call $\cT$
\emph{restricted}, if all nodes 
of depth at most $n^2$ are of
of lowest complexity, 
i.e., type (\ref{type:I})
in Definition~\ref{def:opt}.
We show 
that in a \emph{restricted} linear 
comparison tree, each $\cR_v$ for
a node of depth at most $n^2$ is 
the Cartesian product of planar polygons. 
This will enable us to analyze each input point
independently.

\begin{prop}\label{prop:Cartesian}
  Let $\cT$ be a restricted linear 
  comparison tree, and $v$ a node of
  $\cT$ with $d_v \leq n^2$. There exists a sequence 
  $R_1, \dots, R_n$
  of (possibly unbounded) convex planar 
  polygons such that 
  $\cR_v = \prod_{i=1}^n R_i$. 
  That is, the evaluation of $\cT$ on 
  $P = \langle p_1, \dots, p_n \rangle$ 
  reaches $v$ if and only 
  if $p_i \in R_i$ for all $i$.
\end{prop}

\begin{proof}
  We do induction on $d_v$. For the root, 
  set $R_1 = \dots = R_n = \R^2$. If 
  $d_v \geq 1$, let $v'$ be the
  parent of $v$. By induction, there 
  are planar convex polygons $R_i'$ with 
  $\cR_{v'} = \prod_{i=1}^n R_i'$. 
  As $\cT$ is restricted, $v'$ is 
  labeled with a test 
  ``$p_j \in \ell_{v'}^+$?'', 
  the line $\ell_{v'}$ being independent of 
  $P$. We take $R_i = R_i'$ for $i \neq j$, 
  and $R_j = R_j' \cap \ell_{v'}^+$, if $v$ is 
  the left child of $v'$, and 
  $R_j = R_j' \cap \ell_{v'}^-$, otherwise.
\end{proof}

Next, we restrict linear comparison 
trees even further, so that the depth 
of a node $v$ relates to the probability 
that $v$ is reached by a random input 
$P \sim \cD$. This allows us to
compare the expected running time of our algorithms 
with the depth of a near optimal tree.

\begin{definition}
  A restricted 
  comparison tree is 
  \emph{entropy-sensitive} if the 
  following holds for any node $v$
  with $d_v \leq n^2$:
  let $\cR_v = \prod_{i=1}^{n} R_i$
  and $v$ labeled 
  ``$p_j \in \ell_{v}^+$\textup{?}''.
  Then $\ell_v$ is a halving line for $R_j$.
\end{definition}

The depth of a node in an entropy-sensitive
linear comparison tree is related to the 
probability that it is being visited:

\begin{prop}\label{prop:entropyDepth}
  Let $v$ be a node in an entropy-sensitive 
  tree with $d_v \leq n^2$, and 
  $\cR_v = \prod_{i=1}^n R_i$.  Then, 
  \[
    d_v =  -\sum_{i=1}^n \log \Pr_{p_i \sim \cD_i}[p_i \in R_i].
  \]
\end{prop}

\begin{proof}
  We do induction on $d_v$. 
  The root has depth $0$ and all 
  probabilities are $1$. The claim 
  holds.  Now let $d_v \geq 1$ and 
  $v'$ be the parent of $v$. 
  Write $\cR_{v'} = \prod_{i=1}^n R_i'$ 
  and $\cR_{v} = \prod_{i=1}^n R_i$. 
  By induction, 
  $d_{v'} =  -\sum_{i=1}^n \log \Pr[p_i \in R'_i]$. 
  Since $\cT$ is entropy-sensitive, $v'$ is 
  labeled  ``$p_j \in \ell_{v'}^+$?'', 
  where  $\ell_{v'}$ is a halving line
  for $R_j'$, i.e.,
  \[
    \Pr[p_j \in R_j' \cap \ell_{v'}^+] = 
    \Pr[p_j \in R_j' \cap \ell_{v'}^-] = 
    \Pr[p_j \in R_j']/2.
  \]
  Since $R_i = R_i'$, for $i \neq j$, and 
  $R_j = R_j' \cap \ell_{v'}^+$ or 
  $R_j = R_j' \cap \ell_{v'}^-$,
  it follows that
  \[
    -\sum_{i=1}^n \log \Pr_{p_i \sim \cD_i}[p_i \in R_i]
    = 1-\sum_{i=1}^n \log \Pr_{p_i \sim \cD_i}[p_i \in R'_i]
    = 1+d_{v'} = d_v.
  \]
\end{proof}

We prove that it suffices to restrict our attention to entropy-sensitive comparison trees. 
The following lemma is crucial to the proof, as it gives handles on $\OPTMAX$ and $\OPTCH$.

\begin{lemma}\label{lem:lin->entropy}
  Let $\cT$ be a finite linear comparison 
  tree of worst-case depth $n^2$, and $\cD$ a product distribution over 
  points. There is an entropy-sensitive 
  tree $\cT'$ with expected depth 
  $d_{\cD}(\cT') = O(d_\cD(\cT))$, as 
  $d_\cD(\cT) \rightarrow\infty$.
\end{lemma}

This is proven by converting $\cT$ 
to an entropy-sensitive comparison 
tree whose expected depth is only 
a constant factor worse.  This is done in two steps. The first, more technical
step (Lemma~\ref{lem:lin->restricted}), goes from linear comparison trees to restricted comparison trees.
The second step goes from restricted comparison trees to entropy-sensitive trees (Lemma~\ref{lem:restricted->entropy}).
Lemma~\ref{lem:lin->entropy} follows immediately from Lemmas~\ref{lem:lin->restricted} and~\ref{lem:restricted->entropy}.

\begin{lemma}\label{lem:lin->restricted}
  Let $\cT$ be a linear 
  comparison tree of worst-case
  depth $n^2$ and $\cD$ a product
  distribution. There is a restricted 
  comparison tree $\cT'$ with expected depth 
  $d_{\cD}(\cT') = O(d_\cD(\cT))$, as 
  $d_\cD(\cT) \rightarrow\infty$.
\end{lemma}

\begin{lemma}\label{lem:restricted->entropy}
  Let $\cT$ a restricted linear comparison 
  tree. 
  There exists an entropy-sensitive comparison
  tree $\cT'$ with expected depth $d_{\cT'} = O(d_\cT)$.
\end{lemma}

For convenience, we move the proofs to a separate subsection.

\subsection{Proof of Lemmas~\ref{lem:lin->restricted} and~\ref{lem:restricted->entropy}}
\label{sec:lin->restricted}

The heavy lifting is done by representing a single 
comparison by a restricted linear comparison 
tree, provided that $P$ is drawn from a 
product distribution. The final transformation simply
replaces each node of $\cT$ by the subtree given 
by the next claim. For brevity, we omit the 
subscript $\cD$ from $d_\cD$. 

\begin{claim} \label{clm:node} 
  Consider a comparison $C$ as 
  in Definition~\ref{def:opt}. 
  Let $\cD'$ be a product distribution 
  for $P$ with each $p_i$ drawn 
  from a polygonal region $R_i$. If $C$ is not of
  type \textup(\ref{type:I}\textup), there is
  a restricted linear comparison tree $\cT'_C$ 
  that resolves $C$ with expected depth 
  $O(1)$ (over $\cD'$) and worst-case
  depth $O(n^2)$.
\end{claim}
\begin{figure}
  \centering
  \includegraphics{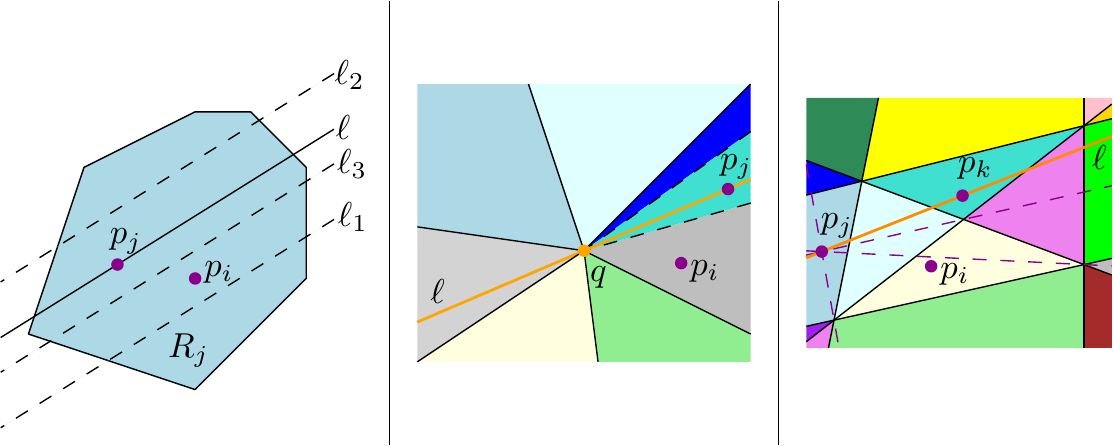}
  \caption{The different cases in the proof of Claim~\ref{clm:node}.}
  \label{fig:compreplace}
\end{figure}
\begin{proof}
We distinguish several cases 
according to Definition~\ref{def:opt}; 
see Fig.~\ref{fig:compreplace}.

\noindent\textbf{$v$ is of type (\ref{type:II}).} 
We must determine 
whether the input point $p_i$ lies 
to the left of the directed line 
with slope $a$ 
through the input $p_j$.
This is done through binary search. 
Let $R_j$ be the region in 
$\cD'$ corresponding to $p_j$,  
and $\ell_1$ a halving line for 
$R_j$ with slope~$a$. We do 
two comparisons to determine 
on which side of $\ell_1$ the inputs 
$p_i$ and $p_j$ lie. If they
lie on different sides, 
we can resolve the original comparison.
If not, we replace 
$R_j$ with the new 
region and repeat. 
Every time, the success 
probability is at least $1/4$.
As soon as the depth exceeds
$n^2$, we use the original
type (\ref{type:II}) comparison.
The probability of reaching
a node of depth $k$ is 
$2^{-\Omega(k)}$, so the expected 
depth is $O(1)$. 
   
\noindent\textbf{$v$ is of type (\ref{type:III}).} 
We must determine whether 
the input point $p_i$ lies 
to the left of the directed line 
through the input $p_j$ 
and the fixed point $q$.
We partition the plane by a 
constant-sized family of cones
with apex $q$, such that 
for each cone $V$ in the family,
the probability that line 
$\overline{q p_j}$ meets $V$ 
(other than at $q$)
is at most $1/2$. Such 
a family can be 
constructed by a sweeping a 
line around $q$, or by
taking a sufficiently large, but 
constant-sized, sample from the 
distribution of $p_j$, and 
bounding the cones by all lines 
through $q$ and each point of 
the sample. As such a construction 
has a positive success probability,
the described family
of cones exists.
  
We build a restricted tree 
that locates a point in the 
corresponding cone, and
for each cone $V$, we recursively 
build such a family of cones 
inside $V$, together with a restricted
tree. Repeating for 
each cone, this gives an infinite 
restricted tree $\cT'_C$. We search 
for both $p_i$ and $p_j$ in $\cT'_C$. 
Once we locate them in 
two different cones of the same 
family, the comparison is resolved.
This happens with probability at least
$1/2$, so 
the probability that the evaluation 
needs $k$ steps is 
$2^{-\Omega(k)}$. Again, we revert
to the original comparison once the
depth exceeds $n^2$.

\noindent\textbf{$v$ is of type (\ref{type:IV}).} 
We must determine whether the 
input point $p_i$ lies to the 
left of the directed line 
through inputs $p_j$ and $p_k$.
We partition the plane by a 
constant-sized family of triangles 
and cones, such that for each
region $V$ in the family, the 
probability that the line 
$\overline{p_jp_k}$ meets $V$ is at 
most $1/2$. Such a family can 
be constructed by taking a 
sufficiently large random sample 
of pairs $p_j$, $p_k$
and by triangulating the arrangement 
of the lines through each pair. The
construction has positive success 
probability, so such a family exists. 
(Other than the source of the random 
lines, this 
scheme goes back at least to \cite{Clarkson87}; 
a tighter version, called \emph{cutting}, 
could also be used \cite{Chazelle93}.) 
 
Now suppose $p_i$ is in region $V$ of 
the family. If the line $\overline{p_jp_k}$ 
does not meet $V$, the comparison
is resolved. This occurs with probability 
at least $1/2$. Moreover, finding the 
region containing $p_i$ takes a 
constant number of type (\ref{type:I})
comparisons.  Determining
if $\overline{p_jp_k}$ meets $V$ can be done 
with a constant number of
type (\ref{type:III}) comparisons: 
suppose $V$ is a triangle. 
If $p_j\in V$, then $\overline{p_jp_k}$ 
meets $V$. Otherwise, suppose $p_k$ 
is above all lines through $p_j$ and each
vertex of $V$; then $\overline{p_jp_k}$ does not 
meet $V$. Also, if $p_k$
is below all lines through $p_j$ and 
each vertex, then $\overline{p_j p_k}$
does not meet $V$. Otherwise, $\overline{p_j p_k}$ 
meets $V$. We replace each type (\ref{type:III})
query by a type (\ref{type:I})
tree, cutting off after $n^2$ levels.
 
By recursively building a tree for each
region $V$ of the family, comparisons of 
type (\ref{type:IV}) can be reduced to
a tree of depth $n^2 + 1$
whose nodes of depth at most $n^2$ 
use comparisons of 
type (\ref{type:I}) only. Since 
the probability of 
resolving the comparison $\Omega(1)$
with each family of regions that 
is visited, the expected
number of nodes visited is constant.
\end{proof}

Given Claim~\ref{clm:node}, we can 
now prove Lemma~\ref{lem:lin->restricted}.

\begin{proof}[Proof of Lemma~\ref{lem:lin->restricted}]
We incrementally transform 
$\cT$ into $\cT'$.
In each step, we have 
a partial restricted 
comparison tree $\cT''$ 
that eventually becomes $\cT'$.
Furthermore, during the 
process each node of 
$\cT$ is in one of 
three different states:
\emph{finished}, \emph{fringe}, 
or \emph{untouched}. 
We also have a function $S$
that assigns to each finished and 
fringe node of $\cT$ a subset 
$S(v)$ of nodes in $\cT''$.
The initial situation is 
as follows: all nodes of $\cT$ 
are untouched except for the 
root, which is fringe. 
The partial tree 
$\cT''$ has a single 
root node $r$, and the function
$S$ assigns the root of $\cT$ 
to the set  $\{r\}$.

The transformation proceeds as 
follows: we pick a fringe node 
$v$ in $\cT$, and mark it as finished.
For each child $v'$ of $v$, if $v'$ 
is an internal node of $\cT$, we mark 
it as fringe. Otherwise, we mark 
$v'$ as finished.
For each node $w \in S(v)$,
if $w$ has depth more than $n^2$,
we copy the subtree of $v$ in
$\cT$ to a subtree of $w$ in
$\cT''$.
Otherwise, 
we replace $w$  
by the subtree given 
by Claim~\ref{clm:node}.  This
is a valid application of the claim, 
since $w$ is a node of $\cT''$, a 
restricted tree. 
Hence $\cR_w$ is a 
product set, and the distribution
$\cD$ restricted to $\cR_w$ is 
a product distribution.
Now $S(v)$ contains the roots of 
these subtrees. Each leaf of each such
subtree corresponds to an outcome of 
the comparison in $v$.  
For each child $v'$ of $v$, we define 
$S(v')$ as the set of all such leaves 
that correspond to the same outcome 
of the comparison as $v'$. 
We continue this process until 
there are no fringe nodes left. 
By construction, the resulting 
tree $\cT'$ is restricted. 

It remains to argue that 
$d_{\cT'} = O(d_{\cT})$. 
Let $v$ be a node of $\cT$. 
We define two random variables 
$X_v$ and $Y_v$: $X_v$ is 
the indicator random variable 
for the event that the node
$v$ is traversed for a random 
input $P \sim \cD$.  
The variable $Y_v$ denotes 
the number of nodes traversed 
in $\cT'$ that correspond to $v$ 
(i.e., the number of nodes
needed to simulate the comparison 
at $v$, if it occurs).
We have 
$d_{\cT} = \sum_{v \in \cT} \EX[X_v]$, 
because if the leaf corresponding
input $P \sim \cD$ has depth $d$, exactly 
$d$ nodes are traversed to reach it.
We also have 
$d_{\cT'} = \sum_{v \in \cT} \EX[Y_v]$, 
since each node in $\cT'$ corresponds 
to exactly one node $v$ in $\cT$. 
Claim~\ref{clm:YvBound} below 
shows that $\EX[Y_v] = O(\EX[X_v])$, 
completing the proof.
\end{proof}

\begin{claim}\label{clm:YvBound} 
  $\EX[Y_v] \leq c\EX[X_v]$
\end{claim}

\begin{proof}
Note that 
$\EX[X_v] = \Pr[X_v = 1] = \Pr[P \in \cR_v]$.
Since the sets $\cR_w$, $w \in S(v)$, 
partition $\cR_v$, 
we can write $\EX[Y_v]$ as
\[
  \EX[Y_v \mid X_v = 0]\Pr[X_v = 0] +
  \sum_{w \in S(v)} \EX[Y_v \mid P \in \cR_w]\Pr[P \in \cR_w].
\]
Since $Y_v = 0$ if $P \notin \cR_v$, 
we have $\EX[Y_v \mid X_v = 0] = 0$. 
Also, $\Pr[P \in \cR_v] = 
\sum_{w \in S(v)} \Pr[P \in \cR_w]$. 
Furthermore, by Claim~\ref{clm:node}, 
we have $\EX[Y_v \mid P \in \cR_w] \leq c$.
The claim follows.
\end{proof}

Lemma~\ref{lem:restricted->entropy} is proven using a similar construction.

\begin{proof} (of Lemma~\ref{lem:restricted->entropy}) The original tree is
restricted, so all queries are of the form $p_i \in \ell^+?$, where $\ell^+$
only depends on the current node. Our aim is to only have queries with halving lines.
Similar to the 
reduction for type (\ref{type:II}) 
comparisons in Claim~\ref{clm:node}, 
we use binary search: 
let $\ell_1$ be a halving line
for $R_i$ parallel to $\ell$. We 
compare $p_i$ with $\ell_1$. If 
this resolves the original comparison, 	
we are done. If not, we 
repeat with the halving 
line for the new region $R_i'$
stopping after $n^2$ steps. 
In each step, the success probability
is at least $1/2$, so the resulting 
comparison tree has constant expected 
depth. We apply the construction
of Lemma~\ref{lem:lin->restricted} 
to argue that for a restricted 
tree $\cT$ there is an entropy-sensitive 
version $\cT'$ whose
expected depth is higher by at 
most a constant factor.
\end{proof}

\section{Search trees and restricted searches} \label{sec:rest}

We introduce the central notion of 
\emph{restricted searches}. For this
we use the following more abstract setting:
let $\bU$ be an ordered finite set 
and $\cF$ be a distribution over $\bU$ 
that assigns each element $j \in \bU$, 
a probability $q(j)$.
Given a sequence $\{a(j) | {j \in \bU}\}$ of numbers 
and an interval $S \subseteq \bU$, we write $a(S)$ 
for $\sum_{j \in S} a(j)$.
Thus, if $S$ is an interval of $\bU$, 
then $q(S)$ is the total probability of $S$.

\begin{figure}
  \centering
  \includegraphics{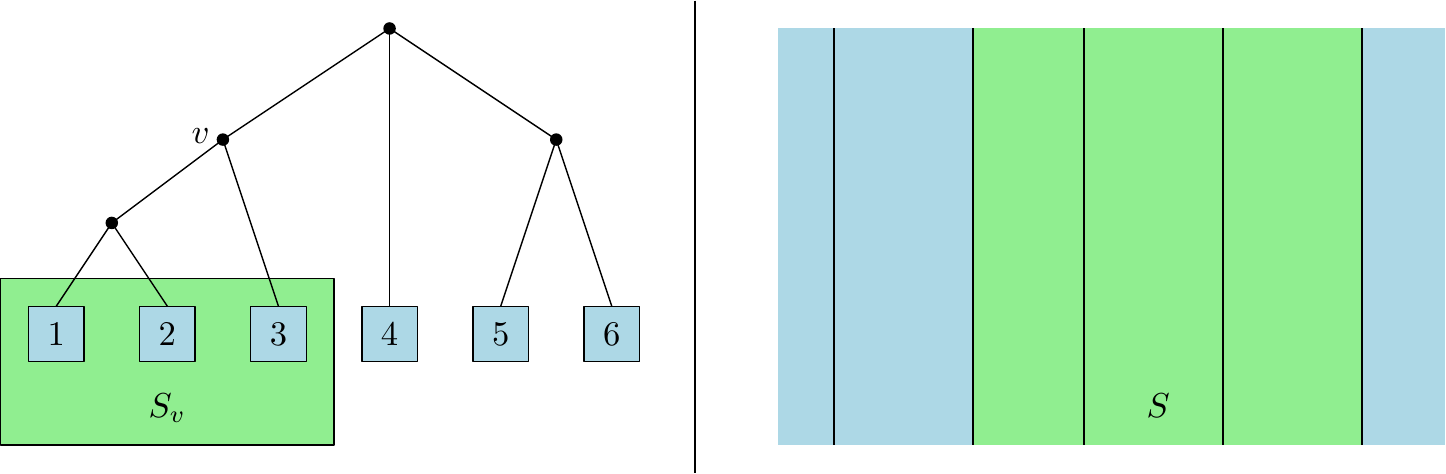}
  \caption{(left) A universe of size $6$ and a 
  search tree. The nodes are ternary, 
  with at most two internal children.
  Node $v$ represents the interval 
  $S_v = \{1,2,3\}$.
  (right) A vertical slab structure with 
  $6$ leaf slabs (including the left and
  right unbounded slab). 
  $S$ is a  slab with
  $3$ leaf slabs, $|S| = 3$.  }
 \label{fig:structures}
\end{figure}

Let $T$ be a search tree over $\bU$. We think 
of $T$ as (at most) ternary, each node having at 
most two internal nodes as children.
Each internal node $v$ of $T$ is associated
with an interval $S_v \subseteq \bU$ so that
every element in $S_v$ has $v$ on its 
search path; see Fig.~\ref{fig:structures}. 
In our setting, $\bU$ is 
the set of leaf slabs of a slab structure 
$\bS$; see Section~\ref{sec:data-str}.
We now define restricted searches.

\begin{definition}\label{def:rest}
  Let $S \subseteq \bU$  be an interval.
  An \emph{$S$-restricted distribution} 
  $\cF_S$ assigns to each $j \in \bU$
  the probability $\xi(j)/\sum_{r \in \bU} \xi(r)$, 
  where $\xi(j)$  fulfills $0 \leq \xi(j) \leq q(j)$, 
  if $j \in S$; and $\xi(j) = 0$, otherwise.  

  An \emph{$S$-restricted search} for $j \in S$ 
  is a search for $j$ in $T$ that terminates 
  as soon as it reaches the first node $v$ with 
  $S_v \subseteq S$.
\end{definition}

\begin{definition}\label{def:tree}
  Let $\mu \in (0,1)$.  A search tree 
  $T$ over $\bU$ is \emph{$\mu$-reducing} 
  if for any internal node $v$ and 
  for any non-leaf child $w$ of $v$, we 
  have $q(S_w) \leq \mu \cdot q(S_v)$. 

  The tree $T$ is \emph{$\alpha$-optimal 
  for restricted searches over $\cF$} if 
  for every interval $S \subseteq \bU$ 
  and every $S$-restricted distribution 
  $\cF_S$, the expected time of an $S$-restricted 
  search over $\cF_S$ is at most $\alpha(1-\log \xi(S))$. 
  (The values $\xi(j)$ are as in Definition~\ref{def:rest}.)
\end{definition}

Our main lemma states that a
search tree that is near-optimal for 
$\cF$ also works for
restricted distributions. 

\begin{lemma}\label{lem:search-time} 
  Let $T$ be a $\mu$-reducing search tree 
  for $\cF$. Then $T$ is $O(1/\log(1/\mu))$-optimal 
  for restricted searches over $\cF$.
\end{lemma}

\subsection{Proof of Lemma~\ref{lem:search-time}}\label{sec:restricted}

We bound the expected number 
of visited nodes in an 
$S$-restricted search. 
Let $v$ be a node of $T$.  
In the following, we use $q_v$ 
and $\xi_v$ as 
a shorthand for the values 
$q(S_v)$ and $\xi(S_v)$. 
Let $\vis(v)$ be the expected number 
of nodes visited below $v$, conditioned
on $v$ being visited.
We prove below, by induction 
on the height of $v$, that for all
visited nodes $v$ with $q_v \leq 1/2$, 
\begin{equation}\label{equ:st_induction}
  \vis(v) \leq c_1 + c\log(q_{v}/\xi_v),
\end{equation}
for some constants $c,c_1 > 0$.  

Given (\ref{equ:st_induction}), 
the lemma follows easily:
since $T$ is $\mu$-reducing, 
for $v$ at depth $k$, we 
have $q_{v} \leq \mu^k$. Hence, 
we have $q_{v} \le 1/2$ for 
all but the root and at most 
$1/\log(1/\mu)$ nodes below 
it (at each
level of $T$ there can be at most 
one node with $q_v > 1/2$).
Let $W$ be the set of nodes $w$ of $T$ 
such that $q_w \leq 1/2$, but
$q_{w'} > 1/2$, for the parent $w'$ 
of $w$. Since $T$ has bounded degree,
$|W| = O(1/\log(1/\mu))$. The expected 
number $\vis(T)$ of nodes visited 
in an $S$-restricted search is at most
\begin{align*}
  \vis(T) &\leq 
    1/\log(1/\mu) + \sum_{w \in W} \Pr_{\cF_S}[j \in S_w]\vis(w)\\ 
  &\leq 
    1/\log(1/\mu) + c_1 + 
      c\sum_{w \in W} \Pr_{\cF_S}[j \in S_w]\log(q_w/\xi_w)\\
  &\leq 
    1/\log(1/\mu) + c_1 + c\sum_{w \in W} \Pr_{\cF_S}[j \in S_w]\log(1/\xi_w),
  \intertext{using (\ref{equ:st_induction}) 
    and $q_w \leq 1$. By definition of
    $\cF_S$, we have 
    $\Pr_{\cF_S}[j \in S_w] = \xi(S_w) / \xi(S)$ $(= \xi_w/\xi(S))$, so}
  \vis(T) &\leq
      1/\log(1/\mu) + c_1 + 
      c\sum_{w \in W} \frac{\xi_w}{\xi(S)}\log(1/\xi_w)\\
    &=
      1/\log(1/\mu) + c_1 + 
      c\sum_{w \in W} \frac{\xi_w}{\xi(S)}(\log(\xi(S)/\xi_w) - \log \xi(S)).\\
   \intertext{The sum 
     $\sum_{w \in W} (\xi_w/\xi(S))\log(\xi(S)/\xi_w)$ represents
     the entropy of a distribution over $W$. 
     Hence, it is bounded by $\log |W|$. 
     Furthermore, $\sum_{w \in W} \xi_w \leq \xi(S)$, so }
    \vis(T) &\leq 
      1/\log(1/\mu) + c_1 + \log|W| - c \log \xi(S) = O(1-\log \xi(S)).
\end{align*}

It remains to prove 
(\ref{equ:st_induction}).
For this, we examine the 
paths in $T$ that an 
$S$-restricted search can lead to.
It will be helpful to consider 
the possible ways how $S$  
intersects the intervals 
corresponding to the nodes 
visited in a search. The 
intersection $S \cap S_v$ of 
$S$ with interval $S_v$ is 
\emph{trivial} if it is 
either empty, $S$, or $S_v$. 
It is \emph{anchored} if it 
shares at least one boundary 
line with $S$.  If 
$S \cap S_v = S_v$, the search 
terminates at $v$, since 
we have certified that $j \in S$. 
If $S \cap S_v = S$, 
then $S$ is contained in $S_v$.
There can be at most one child of 
$v$ that contains $S$. If such a 
child exists, the search continues 
to this child. If not, all 
possible children (to which the search can proceed to) 
are anchored. The search
can continue to any child, at most two 
of which are internal nodes.
If $S_v$ is anchored,  at most one child 
of $v$ can be anchored with $S$.
Any other child that intersects $S$ must 
be contained in it; see
Fig.~\ref{fig:anchored}.

\begin{figure}
\centering
\includegraphics{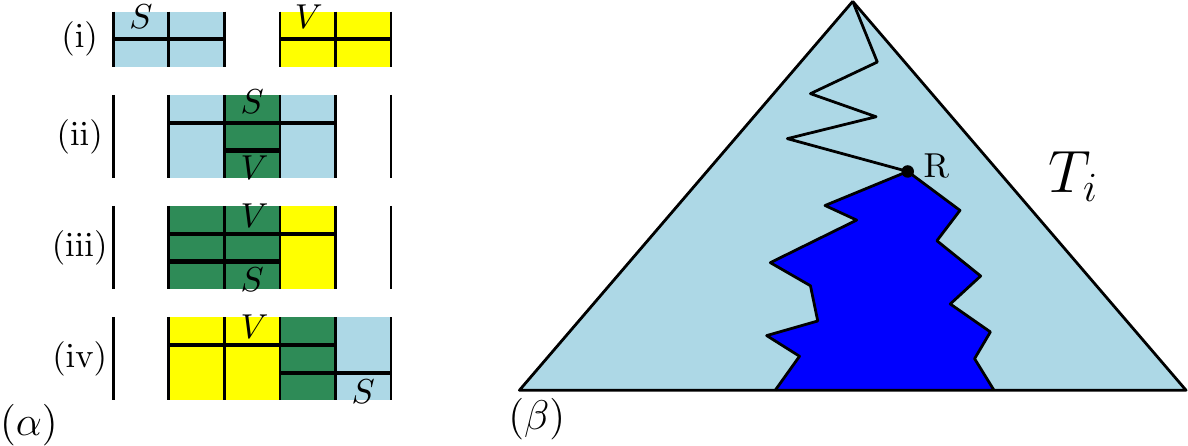}
\caption{($\alpha$) The 
intersections $S \cap S_v$ in (i)-(iii) 
are trivial, the
intersections in (iii) and (iv) 
are anchored; ($\beta$) every 
node of $T_i$
has at most one non-trivial child, 
except for $r$.}
\label{fig:anchored}
\end{figure}

Consider all nodes that 
can be visited by an $S$-restricted 
search (remove all nodes that are 
terminal, i.e., completely 
contained in $S$).  They form a 
set of paths, inducing a subtree 
of $S$. In this subtree, there 
is at most one node with two children. 
This comes from some node $r$ 
that contains $S$ and has two 
anchored (non-leaf) children. Every 
other node of the subtree has a 
single child; see Fig.~\ref{fig:anchored}.
We now prove two lemmas.

\begin{claim}\label{clm:WA}
  Let $v \neq r$ be a 
  non-terminal node that 
  can be visited by 
  an $S$-restricted
  search, and let $w$ 
  be the unique non-terminal 
  child of $v$. Suppose $q_v \le 1/2$
  and $\vis(w) \le c_1 + c\log(q_w/\xi_w)$.
  Then, for $c\ge c_1/\log(1/\mu)$, we have
  \begin{align} \label{T_i recur}
   \vis(v) \le 1 + c\log(q_v/\xi_v).
  \end{align}
\end{claim}

\begin{proof}
From the fact that when a 
search for $j$ shows that 
it is contained in a node 
contained in $S$, the 
$S$-restricted search 
is complete, it follows that
\begin{equation}\label{equ:vis-recursion}
  \vis(v) \leq 
    1 + \frac{\Pr_{\cF_S}[j \in S_w]}{\Pr_{\cF_S}[j \in S_v]}\vis(w) 
  = 1 + \frac{\xi_w}{\xi_v}\vis(w).
\end{equation}
Using the hypothesis, if follows that 
\begin{align*}
  \vis(v) &\leq 1 + \frac{\xi_w}{\xi_v}(c_1 + c\log (q_w/\xi_w)).\\
  \intertext{Since $q_w \le \mu q_v$, 
  and letting $\beta := \xi_w / \xi_v \leq 1$,
  this is}
  &\leq  1 + \beta(c_1 + c\log (q_v/\xi_w) + c\log \mu)\\ 
  &=  1 + \beta c_1 + \beta c \log q_v
   + \beta c \log(1/\xi_w) + \beta c \log \mu.
\end{align*}
The function $x \mapsto x\log(1/x)$ 
is increasing for $x \in (0,1/2)$, 
so $\xi_w\log(1/\xi_w)\leq \xi_v\log(1/\xi_v)$ 
for  $\xi_v\leq q_v\leq 1/2$.
Together with $\beta = \xi_w / \xi_v \le 1$, 
this implies
\begin{align*}
\vis(v) &\le  1 + \beta c_1 + c\log q_v 
  + c\log(1/\xi_v) + \beta c \log \mu \\
   &=  1 + c\log (q_v/\xi_v) + \beta(c_1 + c\log \mu ) 
   \leq  1 + c\log (q_v/\xi_v),
\end{align*}
for $c \ge c_1/\log(1/\mu)$.
\end{proof}

Only a slightly weaker statement 
can be made for the node $r$ having two
nontrivial intersections at child 
nodes $r_1$ and $r_2$.

\begin{claim}\label{clm:WR}
  Let $r$ be as above, and 
  let $r_1,r_2$ be the 
  two non-terminal children of
  $r$. Suppose that 
  $\vis(r_i)\le c_1 + 
  c\log (q_{r_i}/\xi_{r_i})$, 
  for $i=1,2$. Then, 
  for $c \ge c_1/\log(1/\mu)$, 
  we have
  \[
    \vis(r) \le 1 + c\log(q_r/\xi_r) + c.
  \]
\end{claim}

\begin{proof}
Similar to (\ref{equ:vis-recursion}), 
we get
\[
  \vis(r) \leq 
  1 + \frac{\xi_{r_1}}{\xi_r}\vis(r_1) + \frac{\xi_{r_2}}{\xi_r}\vis(r_2).
\]
Applying the hypothesis, 
we conclude
\[
\vis(r) \leq 
  1 + \sum_{i=1}^2 \frac{\xi_{r_i}}{\xi_r}[c_1 + c\log (q_{r_i}/\xi_{r_i})].
\]
Setting $\beta := (\xi_{r_1} + \xi_{r_2})/\xi_r$ 
and using $q_{r_i}\le \mu q_r$, we get 
\begin{align*}
  \vis(r) &\leq  
    1 + \beta c_1 + \beta c \log \mu  + \beta c \log q_r  
    + c \sum_{i=1}^2 (\xi_{r_i}/\xi_r)\log(1/\xi_{r_i}).\\
\intertext{The sum is 
  maximized for $\xi_{r_1} = \xi_{r_2} = \xi_r/2$, 
  so using once again that $\beta \leq 1$, it follows that}
  \vis(r) &\leq 
  1 + \beta c_1 + \beta c \log \mu  + \beta c \log q_r  + c \log(2/\xi_r)\\
  &\leq 
    1 + \beta(c_1 + c\log \mu ) + c\log(q_r/\xi_r) + c\log 2 \\
  &\leq 
    1 + c\log(q_r/\xi_r) + c,
\end{align*}
for $c \geq c_1/\log(1/\mu)$, 
as in (\ref{T_i recur}), except 
for the addition of $c$.
\end{proof}

Now we use Claims~\ref{clm:WA} 
and~\ref{clm:WR} to prove 
(\ref{equ:st_induction}) by induction. 
The bound clearly holds for leaves.
For the visited nodes
below $r$, we may inductively 
take $c_1 = 1$ and $c=1/\log(1/\mu)$, by 
Claim~\ref{clm:WA}.
We then apply Claim~\ref{clm:WR} for $r$.
For the parent $v$ of $r$, we 
use Claim~\ref{clm:WA} with
$c_1 = 1 + 1/\log(1/\mu)$ 
and $c\ge c_1/\log(1/\mu)$, getting
$\vis(v)\le 1 + c\log(q_v/\xi_v)$. 
Repeated application of Claim~\ref{clm:WA}  
(with the given value of $c$) gives 
that this bound also holds for the
ancestors of $v$, at least up until 
the $1+1/\log(1/\mu)$ top nodes. 
This finishes the proof of 
(\ref{equ:st_induction}), and 
hence of Lemma~\ref{lem:search-time}.

\section{Auxiliary data structures} \label{sec:data-str}

We start with a simple heap-structure 
that maintains (key, index) pairs.
The indices are distinct elements of 
$[n]$, and the keys come
from the ordered universe $\{1, \dots, U\}$ ($U \leq n$). 
We store the pairs in a data structure
with operations \texttt{insert}, 
\texttt{delete} (deleting a pair),
\texttt{find-max} (finding the maximum 
key among the stored pairs), and 
\texttt{decrease-key}
(decreasing the key of a pair). 
For \texttt{delete} and
\texttt{decrease-key}, we assume the input 
is a pointer into the
data structure to the appropriate pair.

\begin{claim}\label{clm:ds} 
  Suppose there are $x$ \textup{\findmax{}} 
  operations and $y$ \textup{\deckey{}} operations,
  and that all insertions are performed at the beginning.
  We can implement the heap structure such that the total time 
  for all operations is $O(n + x + y)$. The storage requirement is
  $O(n)$.
\end{claim}

\begin{proof} 
  We represent the heap as an array 
  of lists. For every $k \in [U]$,
  we store the list of indices with key 
  $k$. We also maintain $m$, the current maximum 
  key. The total storage is $O(n)$. A \findmax{}  
  takes $O(1)$ time, and \ins{} is 
  done by adding the element to the appropriate 
  list.  To \delete{}, we remove the element from 
  the list (assuming appropriate pointers are available), 
  and we update the maximum. If the list 
  at $m$ is non-empty, no action is required. If it is empty,
  we check sequentially if the list at $m-1, m-2, \ldots$ 
  is empty.  This eventually leads to the maximum. 
  For \deckey{}, we \delete{}, \ins{}, and then 
  update the maximum.
  Since all insertions happen at the start, the maximum 
  can only decrease, and the total overhead for finding new 
  maxima is $O(n)$.
\end{proof}

Our algorithms use several data structures
to guide the searches. A \emph{vertical slab structure} $\bS$ is a 
sequence of vertical lines that
partition the plane into open \emph{leaf slabs}. 
(Since we assume continuous distributions, we 
may ignore the case that an input point lies on
a vertical line and consider the leaf slabs to 
partition the plane.)
More generally, a \emph{slab} is the 
region between any two vertical lines of $\bS$.  
The \emph{size} of a slab $S$, $|S|$, 
is the number of leaf slabs in it. 
The size of $\bS$, $|\bS|$, is the total number 
of leaf slabs.  For any slab $S$, 
the probability that $p_i \sim \cD_i$ 
is in $S$ is denoted by $q(i,S)$.
Our algorithms construct slab 
structures in the learning phase, 
similar to the algorithm 
in~\cite{AilonCCLMS11}.

\begin{lemma}\label{lem:slabstruct}
  We can build a slab structure \textup{$\bS$} 
  with $O(n)$ leaf slabs so that the following holds
  with probability $1-n^{-3}$ over the construction:
  for a leaf slab $\lambda$ of \textup{$\bS$}, 
  let $X_\lambda$  be the number of points in a 
  random input $P$ that lie in $\lambda$.  Then 
  $\EX[X^2_\lambda] = O(1)$, for every leaf slab
  $\lambda$.  The construction 
  takes $O(\log n)$ rounds and $O(n\log^2 n)$ time.
\end{lemma}

\begin{proof} The construction is identical to
the $V$-list in Ailon \etal~\cite[Lemma~3.2]{AilonCCLMS11}: 
take $t = \log n$ random inputs
$P_1, \dots, P_t$, 
and let
$-\infty =: x_0,x_1,\ldots,x_{nt},x_{nt+1} := +\infty$
be the sorted list of the $x$-coordinates
of the points 
(extended by $-\infty$ and $\infty$).
The $n$ values 
$x_0,x_t,x_{2t},\ldots,x_{(n-1)t}$ 
define the boundaries for the slabs in $\bS$. 
Lemma~3.2 in Ailon \etal~\cite{AilonCCLMS11}
shows that  for each leaf slab $\lambda$
of $\bS$, the number $X_\lambda$ of points
in a random input $P$ that lie in $\lambda$
has $\EX_\cD[X_\lambda] = O(1)$ and 
$\EX_\cD[X_\lambda^2] = O(1)$, with 
probability at least $1-n^{-3}$ over
the construction of $\bS$.
The proof is completed by noting that sorting the
$t$ inputs $P_1, \dots, P_t$ takes 
$O(n \log^2 n)$ time.
\end{proof}

The algorithms construct a  
specialized search tree on $\bS$ for 
each distribution $\cD_i$.
It is important to store these
trees with little space.  
The following lemma gives the details
the construction.

\begin{lemma}\label{lem:tree} 
  Let $\eps > 0$ be fixed and
  $\bS$ a slab structure with $O(n)$ 
  leaf slabs. In $O(n^{\eps})$ rounds 
  and $O(n^{1+\eps})$ time, we can 
  construct search trees $T_1, \dots, T_n$ 
  over $\bS$ such that the following holds:
  \textup(i\textup) the trees can be 
  need $O(n^{1+\eps})$ total space;
  \textup(ii\textup) with probability $1-n^{-3}$ 
  over the construction, 
  each $T_i$ is $O(1/\eps)$-optimal for 
  restricted searches over $\cD_i$.
\end{lemma}

Once $\bS$ is constructed,
the search trees $T_i$ can 
be found using essentially 
the same techniques in 
Ailon \etal~\cite[Section~3.2]{AilonCCLMS11}:
we use $n^\eps\log n$ rounds to 
build the first
$\eps \log n$ levels of each $T_i$, 
and we use a balanced search tree 
for searches that proceed to a deeper 
level. This only costs a factor of $1/\eps$.
The proof of Lemma~\ref{lem:tree} is 
almost the same as that in 
Ailon \etal~\cite[Section~3.2]{AilonCCLMS11},
but since we require the additional
property of restricted search optimality,
we redo it for our setting. 

\subsection{Proof of Lemma~\ref{lem:tree}}
\label{sec:learning}

Let $\delta > 0$ be a sufficiently 
small constant and $c > 0$
be sufficiently large. 
We take $t = c\delta^{-2}n^{\eps}\log n$ 
random inputs, and  for each $p_i$, we 
record the leaf slab of $\bS$ that 
contains it. We break the
proof into smaller claims.

\begin{claim} \label{clm:prob} 
  Using $t$ inputs, we can
  obtain estimates 
  $\hat{q}(i,S)$ for each input point $p_i$ 
  and each slab $S$ such that following 
  holds (for all $i$ and $S$)
  with probability at least $1 - n^{-3}$ over 
  the construction: 
  if at least $(c/10e\delta^{2})\log n$ 
  instances of $p_i$ fell in $S$,
  then $\hat{q}(i,S) \in  
  [(1-\delta)q(i,S),(1+\delta)q(i,S)]$.\footnote{We 
  remind the reader that this the
  probability that $p_i \in S$.}
\end{claim}

\begin{proof}
Fix $p_i$ and $S$, and
let $N(i,S)$ be the 
number of times $p_i$
was in $S$. 
Let $\hat{q}(i,S) = N(i,S)/t$ 
be the empirical 
probability for this event.
$N(i,S)$ is a sum of independent
random variables, and
$\EX[N(i,S)] = tq(i,S)$.
If $\EX[N(i,S)] < (c/10e\delta^{2})\log n$, 
then 
$2e\EX[N(i,S)] < (c/5\delta^2)\log n $, so
by a Chernoff 
bound~\cite[Theorem~1.1, Eq. (1.8)]{DubhashiPa09}, 
\[
  \Pr[N(i,S) > (c/5\delta^2)\log n] \leq 2^{-(c/5\delta^2)\log n} \leq n^{-6}.
\]
Hence, with probability at 
least $1-n^{-6}$, if 
$N(i,S) > (c/5\delta^2)\log n$, 
then $\EX[N(i,S)] \geq (c/10e\delta^{2})\log n$.
If $\EX[N(i,S)] \geq (c/10e\delta^{2})\log n$,  
multiplicative Chernoff 
bounds~\cite[Theorem 1.1, Eq. (1.7)]{DubhashiPa09}
give
\[
  \Pr[N(i,S) \notin [(1-\delta)\EX[N(i,S)], (1+\delta)\EX[N(i,S)]]] < 
    2\exp(-\delta^2\EX[N(i,S)]/3) < n^{-6}.
\]
The proof is completed by taking a union bound over all $i$ and $S$.
\end{proof}

Assume that the event
of Claim~\ref{clm:prob} holds.
If at least $(c/10e\delta^{2})\log n$ 
inputs fell in $S$, then 
$\hat{q}(i,S) = \Omega(n^{-\eps})$
and $q(i,S) = \Omega(n^{-\eps})$.
The tree $T_i$ is constructed recursively.
We first create a partial search tree, where
some leaves may not correspond to leaf slabs.
The root of $T_i$ corresponds to 
$\R^2$. Given a slab
$S$, we proceed as follows:
if $N(S) < (c/10e\delta^{2})\log n$, we make
$S$ a leaf. If not, we pick a leaf slab 
$\lambda$ such that the subslab 
$S_l \subseteq S$ with all leaf slabs 
strictly to the left of $\lambda$ and
the subslab $S_r \subseteq S$ with all 
leaf slabs strictly 
to the right of $\lambda$ have 
$\hat{q}(i,S_l) \leq (2/3)\hat{q}(i,S)$ and
$\hat{q}(i,S_r) \leq (2/3)\hat{q}(i,S)$. 
We make $\lambda$ a leaf child of $S$, and 
we recursively create trees for 
$S_l$ and $S_r$ and attach them to $S$. 
For any internal node $S$, we have 
$q(i,S) = \Omega(n^\eps)$,
so the depth is $O(\eps \log n)$. 
Furthermore,
the partial tree $T_i$ is $\beta$-reducing 
(for some constant $\beta$).
We get a complete tree
by constructing a balanced tree
for each $T_i$-leaf that is not a leaf slab. 
This yields a tree  of depth 
at most $(1+O(\eps))\log n$.
We only need to store the 
partial tree, 
so the total space 
is $O(n^{1+\eps})$.

\begin{claim}\label{clm:1/eps} 
  The tree $T_i$ is $O(1/\eps)$-optimal
  for restricted searches. 
\end{claim}

\begin{proof}
Fix an $S$-restricted
distribution $\cF_S$. For each leaf
slab $\lambda$, let $q'(i,\lambda)$ 
be the probability according to 
$\cF_S$. Note that 
$q'(i,S) \leq q(i,S)$.
If $q'(i,S) \leq n^{-\eps/2}$, 
then $-\log q'(i,S) \geq \eps (\log n)/2$.
Any search in $T_i$ takes at 
most $(1+O(\eps))\log n$ steps,
so the search time is
$O(\eps^{-1}(-\log q'(i,S) + 1))$.

Now suppose $q'(i,S) > n^{-\eps/2}$. 
Consider a search for $p_i$.
We classify the search according to 
the leaf that it reaches in
the partial tree.
By construction, any leaf $S'$
of $T_i$ is either a leaf slab or
has $q(i,S') = O(n^{-\eps})$.
The search is of \emph{Type 1} 
if the leaf of the partial tree 
represents a leaf slab (and 
hence the search terminates). 
The search is of \emph{Type 2} 
(resp. \emph{Type 3}) if the 
leaf of the partial tree 
is an internal node of $T_i$
and the depth is at least 
(resp.~less than) $\eps(\log n)/3$. 

As a thought experiment,
we construct a related tree 
$T'_i$: start with the partial 
$T_i$, and for every leaf that 
is not a leaf slab, extend it 
using the true 
probabilities $q(i,S)$. That  
is, construct the 
subtree rooted at a new node $S$
in the following manner: pick a leaf 
slab $\lambda$ with
$q(i,S_l) \leq (2/3)q(i,S)$ and
$q(i,S_r) \leq (2/3)q(i,S)$ 
(with $S_l$ and $S_r$ as above). 
This ensures that $T'_i$ is 
$\beta$-reducing.
By Lemma~\ref{lem:search-time}, $T'_i$ 
is $O(1)$-optimal for restricted 
searches over $\cF_i$ (we absorb 
$\beta$ into the $O(1)$).

If the search is of Type 1, 
it is identical in both $T_i$ 
and $T'_i$. If
it is of Type 2, it 
takes at least $\eps(\log n)/3$ 
steps in $T'_i$ and at most
$(1+O(\eps))(\log n)$ steps 
in $T_i$. 
Consider Type 3 searches.
The total number of leaves 
(that are not leaf slabs)
of the partial tree at depth 
less than $\eps(\log n)/3$ is 
at most $n^{\eps/3}$.
The total probability mass 
of $\cF_i$ on such leaves is 
$O(n^{\eps/3}\cdot n^{-\eps}) < O(n^{-2\eps/3})$.
Since $q'(i,S) > n^{-\eps/2}$,
the probability of a Type 3 
search is at most $O(n^{-\eps/6})$.

Choose a random $p_i \sim \cF_S$.
Let $\cE$ be the event that 
a Type 3 search occurs. Furthermore, 
let $Z$ be the depth of the search 
in $T_i$ and $Z'$ be
the depth in $T'_i$. 
If $\cE$ does not occur, 
we have argued that $Z = O(Z'/\eps)$.
Also, $\Pr(\cE) = O(n^{-\eps/6})$. 
The expected search time is $\EX[Z]$.
Hence,
\begin{align*}
  \EX[Z] =  
    \Pr[\overline{\cE}] \EX_{\overline{\cE}}[Z] + 
    \Pr[\cE]\EX_{\cE}[Z] 
  &\leq  
    O(\eps^{-1}\EX_{\overline{\cE}}[Z']) +  n^{-\eps/6}(1+O(\eps))\log n \\
  &= O(\eps^{-1}\EX_{\overline{\cE}}[Z'] + 1).
\end{align*}	
Since $\Pr[\overline{\cE}] > 1/2$,
$\EX_{\overline{\cE}}[Z'] \leq 
2\Pr[\overline{\cE}] \EX_{\overline{\cE}}[Z'] 
\leq 2\EX[Z']$.
Combining everything, the 
expected search time is 
$O(\eps^{-1}\EX[Z'] + 1)$. 
Since $T'_i$ is $O(1)$-optimal 
for restricted searches, 
$T_i$ is $O(\eps^{-1})$-optimal.
\end{proof}

\section{A self-improving algorithm for coordinate-wise maxima}

We begin with an informal overview.
If $P$ is sorted by $x$-coordinate, we
can do a
right-to-left sweep:
we maintain the maximum $y$-coordinate $Y$ 
seen so far. When a point $p$ is visited,
if $y(p) < Y$, then $p$ is non-maximal, and the point $q$ with
$Y=y(q)$ gives a per-point certificate for $p$.
If $y(p)\ge Y$, then $p$ is maximal. We update $Y$ and put
$p$ at the beginning of the maxima list of $P$.
This suggests the following approach to a 
self-improving algorithm: sort $P$ with a 
self-improving sorter and then do the 
sweep. The sorting algorithm of \cite{AilonCCLMS11} 
works by locating each point of $P$ within 
the slab structure $\bS$ of 
Lemma~\ref{lem:slabstruct} using the trees $T_i$ of
Lemma~\ref{lem:tree}.

As discussed in Section~\ref{sec:results}, 
this does not work. We need another approach:
as a thought experiment, suppose that the maximal
points of $P$ are available, though not in sorted order. 
We locate the maxima in $\bS$ and 
determine their sorted order. 
We can argue that the optimal algorithm must also
(in essence) perform such a search. To find 
per-point certificates for the non-maximal points,
we use the slab structure $\bS$ and the search trees,
proceeding very conservatively.
Consider the search for a point $p$. At any intermediate
stage, $p$ is placed in a slab $S$.
This rough knowledge of $p$'s location 
may be enough to certify 
its non-maximality: 
let $m$ denote the leftmost maximal
point to the right of $S$ (since the 
sorted list of maxima is known, this information
can be easily deduced). We check if $m$ dominates 
$p$. If so, we have a per-point 
certificate, and we terminate the search. 
Otherwise, we continue the search by a single 
step in the search tree for $p$, and we repeat. 

Non-maximal points that are dominated 
by many maximal points should 
have a short search, while points
that are ``nearly'' maximal should 
need more time. 
Thus, this approach
should derive just the ``right" amount 
of information to determine the maxima.
Unfortunately, our thought experiment requires 
that the maxima be known.
This, of course, is too much to ask, and 
due to the strong dependencies, it is not clear
how to determine the maxima before performing 
the searches.

The final algorithm overcomes this 
difficulty by interleaving the searches 
for sorting the points with confirmation
of the maximality of some points, in 
a rough right-to-left order
that is a more elaborate version of the traversal scheme given
above. 
The searches for all points $p_i$ (in their 
respective trees $T_i$) are performed 
``together'', and their order is carefully chosen. 
At any intermediate
stage, each point $p_i$ is located in 
some slab $S_i$, represented by
a node of its search tree. We 
choose a specific point and advance
its search by one step. This choice is 
very important, and is the basis
of optimality. The algorithm is 
described in detail and analyzed in 
Section~\ref{sec:algorithm}.

\subsection{Restricted Maxima Certificates}

We modify the maxima certificate 
from Definition~\ref{def:cert-max} in order
to get easier proofs of optimality.
For this, we need the following observation,
see Fig.~\ref{fig:domregion}.

\begin{prop}\label{prop:dom} 
  Let $\cT$ be a linear comparison
  tree for computing the maxima.
  Let $v$ be a leaf of $\cT$ 
  and $R_i$ be the region 
  associated with non-maximal point 
  $p_i \in P$ in $\cR_v$. 
  There is a region $R_j$
  associated with a maximal point $p_j$ 
  such that every point in $R_j$
  dominates every point in $R_i$.
\end{prop}
\begin{figure}
\centering
\includegraphics{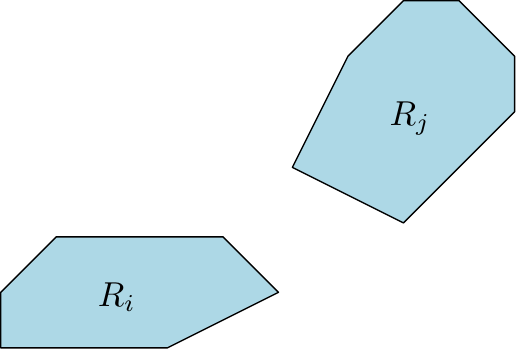}
\caption{Every point in $R_j$ dominates every point in $R_i$.}
\label{fig:domregion}
\end{figure}

\begin{proof} 
The leaf $v$ is associated with a 
certificate $\gamma$ that is valid 
for every input that reaches $v$. 
The certificate $\gamma$ associates 
the non-maximal point $p_i$ with 
$p_j$ such that $p_j$ dominates $p_i$. 
For any input $P$ in $\cR_v$, $p_j$ 
dominates $p_i$. First, we argue 
that $p_j$ can be assumed to be 
maximal. We construct a directed graph 
$G$ with vertex set $[n]$
such that $G$ has an edge $(u,v)$ 
if and only if (according to 
$\gamma$) $p_u$ is dominated by 
$p_v$. All vertices have outdegree at most 
$1$, and there are no cycles in $G$ 
(since dominance is transitive).
Hence, $G$ consists of trees with 
edges directed towards the root.
The roots are maximal vertices, and any 
point in a tree is dominated by the point 
corresponding to the root. We can thus
rewrite $\gamma$ so that all dominating 
points are extremal.

Since $\cT$ is restricted, the 
region $\cR_v \subseteq \R^{2n}$ 
for $v$ is a Cartesian product 
of polygonal regions $R_1, \dots, R_n$.
Suppose there are two points
$p_i \subseteq R_i$ and $p_j \subseteq R_j$
such that $p_j$ does not dominate $p_i$. 
Take an input $P$ where the
remaining points are arbitrarily chosen from their 
respective regions. The certificate $\gamma$
is not valid for $P$, contradicting the 
nature of $\cT$.
Hence, every point in $R_j$ dominates every point 
in $R_i$. 
\end{proof}

We need points 
in the maxima certificate to be 
``well-separated'' according to 
the slab structure $\bS$. 
By Proposition~\ref{prop:dom},
every non-maximal point is associated with a dominating
region.

\begin{definition}\label{def:s-label}
  Let $\bS$ be a slab structure. 
  A maxima certificate for an input $P$ is 
  \emph{$\bS$-labeled} if (i)
  every maximal point is labeled with 
  the leaf slab of $\bS$ containing it; and
  (ii) every non-maximal point is either 
  placed in the containing leaf slab,
  or is separated from its dominating 
  region by a slab boundary.
\end{definition}

A tree $\cT$ \emph{computes the $\bS$-labeled
maxima} if the leaves are associated with 
$\bS$-labeled certificates.

\begin{lemma}\label{lem:comp} 
  There is an entropy-sensitive comparison 
  tree $\cT$ for computing the $\bS$-labeled 
  maxima whose expected depth over $\cD$ is 
  $O(n + \textup{\OPTMAX}_\cD)$.
\end{lemma}

\begin{proof} 
We start with a linear comparison 
tree of depth $O(\OPTMAX_D)$ that computes the maxima, with
certificates as in 
Proposition~\ref{prop:dom}. 
Each leaf has a list $M$ with the maximal
points in sorted order.  We merge $M$ with the 
slab boundaries of $\bS$ to label each 
maximal point with the leaf slab of $\bS$ containing it. 
This needs $O(n)$ additional comparisons.
Now let $R_i$ be the region associated 
with a non-maximal point $p_i$, and $R_j$ 
the maximal dominating region. Let $\lambda$ 
be the leaf slab containing $R_j$.
The $x$-projection of $R_i$ cannot extend 
to the right of $\lambda$.
If there is a slab boundary separating 
$R_i$ from $R_j$, nothing needs to be done.
Otherwise, $R_i$ intersects 
$\lambda$. With one more comparison, we can
place $p_i$ inside $\lambda$ or strictly 
to the left of it.
In total, it takes $O(n)$ additional 
comparisons in each leaf to
that get a tree for the $\bS$-labeled 
maxima. Hence, the expected depth is $O(n + \OPTMAX_\cD)$. 
We apply Lemma~\ref{lem:lin->entropy} 
to get an entropy-sensitive tree with
the desired properties.
\end{proof}

\subsection{The algorithm}
\label{sec:algorithm}

In the learning phase, the algorithm 
constructs a slab structure $\bS$ and
search trees $T_i$ as in 
Lemmas~\ref{lem:slabstruct} and~\ref{lem:tree}.
Henceforth, we assume that we have these structures, 
and we describe the algorithm in the limiting
phase.
The algorithm searches
each point $p_i$ progressively 
in its tree $T_i$, while
interleaving
the searches carefully. 

At any stage of the algorithm, 
each point $p_i$ is placed in some slab 
$S_i$.  The algorithm maintains a set 
$A$ of \emph{active points}. All
other points are either proven 
non-maximal, or placed in 
a leaf slab. 
The heap structure $L(A)$ from
Claim~\ref{clm:ds} is used to
store pairs of indices of active points
and associated keys.
Recall that $L(A)$ supports the operations \ins{},
\delete{}, \deckey{}, and \findmax{}. The key 
for an active point $p_i$ is the  
right boundary of the slab $S_i$ 
(represented as an element of $[|\bS|]$).
We list the variables of the algorithm.
Initially, $A = P$, and each $S_i$ is the largest 
slab in $\bS$.  Hence, all points have  key $|\bS|$, 
and we \ins{} all these pairs
into $L(A)$. 

\begin{asparaenum}
  \item $A, L(A)$: the list $A$ of 
    active points is stored in 
    heap structure $L(A)$, with
    their associated right slab boundary as key.
  \item $\wlambda, B$: Let $m$ be the largest 
     key in $L(A)$. Then $\wlambda$ is the leaf slab 
     with right boundary is $m$ and
     $B$ is a set of points located in $\wlambda$
     so far. Initially $B$ is empty and $m$ is $|S|$, 
     corresponding to the $+\infty$ boundary of the 
     rightmost, infinite, slab.
  \item $M, \hat{p}$: $M$ is a sorted (partial) list 
    of the maximal points so far,
    and $\hat{p}$ is the leftmost among those. 
    Initially $M$ is empty and $\hat{p}$ is a 
    ``null'' point that dominates no input point.
\end{asparaenum}

The algorithm involves a main 
procedure \textbf{Search}, and an 
auxiliary procedure
\textbf{Update}. The procedure \textbf{Search} 
chooses a point and advances
its search by a single node in the corresponding
search tree. Occasionally, \textbf{Search} invokes 
\textbf{Update} to change the global variables. 
The algorithm repeatedly calls \textbf{Search}
until $L(A)$ is empty. 
After that, we make a final call to 
\textbf{Update} in order to
process any remaining points.

\noindent
\textbf{Search}:  
Perform a \findmax{} in $L(A)$
and let $p_i$ be the resulting
point.
If the maximum key $m$ in $L(A)$ is 
less than the right 
boundary of $\widehat{\lambda}$, 
invoke \textbf{Update}. 
If $p_i$ is dominated by $\hat{p}$, 
delete $p_i$ from $L(A)$.
If not, advance 
the search of $p_i$ in $T_i$ by 
a single node, if possible. 
This updates the slab $S_i$. If 
the right boundary of $S_i$ has 
decreased, perform a 
\deckey{} operation on $L(A)$. 
(Otherwise, do nothing.)
Suppose the point $p_i$ reaches 
a leaf slab $\lambda$.  If 
$\lambda = \widehat{\lambda}$, 
remove $p_i$ from $L(A)$ and insert it
in $B$ (in time $O(|B|)$). Otherwise, 
leave $p_i$ in $L(A)$.

\noindent
\textbf{Update}: 
Sort the points in $B$ and update 
the list of maxima.
As Claim~\ref{clm:order} will show, we 
know the sorted list of maxima to 
the right of $\wlambda$.  Hence, 
we can append to this list in $O(|B|)$
time. We reset $B = \emptyset$,
set $\widehat{\lambda}$ to the leaf 
slab to the left of $m$, and return.

The following claim states the main 
important invariant of the algorithm.

\begin{claim}\label{clm:order} 
  At any time in the algorithm, all 
  maxima to the right of $\wlambda$ have 
  been found, in order from right to left.
\end{claim}

\begin{proof} 
The proof is by backward induction on 
$m$, the right boundary of $\wlambda$.
For $m = |S|$, the claim is trivially true. 
Assume it holds for a given value of $m$, 
and trace the algorithm's behavior until 
the maximum key becomes smaller than $m$ 
(which happens in \textbf{Search}). 
When \textbf{Search}
processes a point $p$ with key $m$ 
then either (i) the key value decreases; 
(ii) $p$ is dominated by $\hat{p}$; 
or (iii) $p$ is  placed in 
$\wlambda$ (whose right boundary is $m$). 
In all cases, when the maximum key decreases
below $m$, all points in $\wlambda$ are 
either proven to be non-maximal
or are in $B$. By the induction hypothesis, 
we already have a sorted
list of maxima to the right of $m$.
The procedure \textbf{Update} sorts the 
points in $B$ and all maximal points 
to the right of $m-1$ are determined.
\end{proof}

\subsubsection{Running time analysis}\label{sec:runtime}

We prove the following lemma.

\begin{lemma}\label{lem:algoMaxima} 
  The maxima algorithm runs in 
  $O(n+\textup{\OPTMAX}_\cD)$ time.
\end{lemma}

We can easily bound the running time 
of all calls to \textbf{Update}.

\begin{claim}\label{clm:update} 
  The total expected time for
  calls to \textbf{Update} 
  is $O(n)$.
\end{claim}

\begin{proof}
The total time for the calls to 
\textbf{Update} is at most
the time needed for sorting 
points within each leaf slabs. 
By Lemma~\ref{lem:slabstruct}, 
this takes expected time
\[
  \EX \Bigl[\sum_{\lambda \in \bS} X_\lambda^2\Bigr]
  =  \sum_{\lambda \in \bS} \EX\bigl[X_\lambda^2\bigr]
  =  \sum_{\lambda \in \bS} O(1)
  =  O(n).
\] 
\end{proof}

The following claim is key to relating the time
spent by \textbf{Search} to 
entropy-sensitive comparison trees.

\begin{claim}\label{clm:search} 
  Let $\cT$ be an entropy-sensitive 
  comparison tree computing $\bS$-labeled 
  maxima.  Consider a leaf $v$ with depth
  $d_v \leq n^2$ labeled 
  with the regions $\mathcal{R}_v = 
  R_1 \times \dots \times R_n$. 
  Conditioned on $P \in \cR_v$, the 
  expected running time of 
  \textup{\textbf{Search}} is 
  $O(n + d_v)$.
\end{claim}

\begin{proof} 
For each $R_i$, let $S_i$ be the 
smallest slab of $\bS$ that completely contains 
$R_i$. We will show that the algorithm
performs at most an $S_i$-restricted search 
for input $P \in \cR_v$. 
If $p_i$ is maximal, then $R_i$ is 
contained in a leaf slab (because 
the output is $\bS$-labeled). Hence 
$S_i$ is a leaf slab and an $S_i$-restricted 
search for a maximal $p_i$ is just 
a complete search.

Now consider a non-maximal $p_i$. By the properties 
of $\bS$-labeled maxima, the associated region $R_i$ 
is either inside a leaf slab or
is separated by a slab boundary from the dominating 
region $R_j$.  In the former case, an $S_i$-restricted 
search is a complete search.
In the latter case, an $S_i$-restricted 
search suffices to process $p_i$:  
by Claim~\ref{clm:order}, when an $S_i$-restricted 
search finishes, all maxima to the right of 
$S_i$ have been determined.  In particular, we have 
found $p_j$, so $\hat p$ dominates $p_i$.
Hence, the search for $p_i$ proceeds no further.

The expected search time taken 
conditioned on $P \in \cR_v$ is 
the sum (over $i$) of the conditional 
expected $S_i$-restricted search times. 
Let $\cE_i$ denote the event that 
$p_i \in R_i$, and $\cE$ be the event 
that $P \in \cR_v$.  We have 
$\cE = \bigwedge_i \cE_i$.
By the independence of the distributions 
and linearity of expectation
\begin{align*} 
  \EX_\cE[\text{search time}] 
  &= 
  \sum_{i=1}^n \EX_\cE[\text{$S_i$-restricted search time for $p_i$}] \\
  &=  
  \sum_{i=1}^n \EX_{\cE_i} [\text{$S_i$-restricted search time for $p_i$}].
\end{align*}
By Lemma~\ref{lem:search-time}, the time 
for an $S_i$-restricted search 
conditioned on $p_i \in R_i$ is 
$O(-\log \Pr[p_i \in R_i] + 1)$. 
By Proposition~\ref{prop:entropyDepth}, 
$d_v = \sum_i -\log \Pr[p_i \in R_i]$,
completing the proof.
\end{proof}

We can now prove the main lemma.

\begin{proof}[Proof of Lemma~\ref{lem:algoMaxima}] 
By Lemma~\ref{lem:comp}, there is
an entropy-sensitive 
tree that computes the $\bS$-labeled 
maxima with expected depth $O(\OPTMAX + n)$.
Since the algorithm never exceeds
$O(n^2)$ steps and by Claim~\ref{clm:search}, the 
expected running time of 
\textbf{Search} is $O(\OPTMAX +  n)$, and 
by Claim~\ref{clm:update} the 
total expected time for \textbf{Update}
is $O(n)$. Adding these bounds completes 
the proof.
\end{proof}

\section{A self-improving algorithm for convex hulls}

We outline the main ideas. The basic approach
is the same as for maxima. We set up 
a slab structure $\bS$, and each
distribution has a dedicated tree for searching 
points. At any stage, each point
is at some intermediate node of the search tree, 
and we wish to advance searches for points 
that have the greatest potential for being extremal.
Furthermore, we would like to quickly 
ascertain that a point is not extremal, so that
we can terminate its search.

For maxima, this strategy is easy enough to 
implement. The ``rightmost'' active point 
is a good candidate for being maximal, so we always
proceed its search. The leftmost known maximal
point can be used to obtain certificates of non-maximality.
For convex hulls, this is much more problematic. At any 
stage, there are many points likely to be
extremal, and it is not clear how to choose.
We also need a procedure that can quickly identify
non-extremal points. 

We give a high-level 
description of the main algorithm.
We construct a \emph{canonical hull} 
$\cC$ in the learning phase. The canonical hull
is a crude representative for the actual 
upper hull. The canonical hull has two key 
properties. First, any point that is below 
$\cC$ is likely to be non-extremal. Second, 
there are not too many points above $\cC$. 

The curve $\cC$ is constructed as follows.
For every (upward) direction $v$, take the normal line $\ell_v$
such that the expected total number of points above $\ell_v$
is $\log n$. We can take the intersection of $\ell^-_v$ over all $v$,
to get an upper convex curve $\cC$. Any point below this curve is highly likely to be non-extremal.
Of course, we need a finite description, so we choose some finite set $\textbf{V}$ of directions,
and only consider $\ell^-_v$ for these directions to construct $\cC$. We choose $\textbf{V}$ to
ensure that the expected number of extremal points in the slab corresponding to a segment of $\cC$ is $O(\log n)$.
We build the slab
structure $\bS$ based on these segments of $\cC$, and search for points
in $\bS$. Each search for point $p$ will result in one of the three conclusions:
$p$ is located above $\cC$, $p$ is located below $\cC$, or $p$ is located in a leaf slab.
This procedure is referred to as the \emph{location algorithm}.  

Now, we have some partial information about 
the various points that is used by
a \emph{construction algorithm} to find $\UH(P)$. We can ignore all points below $\cC$,
and prove that the $\UH(P)$ can be found on $O(n\log \log n)$ time.

\subsection{The canonical directions}\label{sec:prelim_CH}

We describes the structures obtained 
in the learning phase.  In order 
to characterize the typical behavior 
of a random input $P \sim \cD$,
we use a set $\textbf{V}$ of 
\emph{canonical directions}. A 
\emph{direction} is a two-dimensional 
unit vector. Directions are ordered 
clockwise, and we only consider 
directions that point upwards. Given 
a direction $v$, we say that 
$p \in P$ is \emph{extremal} for $v$ 
if the scalar product 
$\langle p,v \rangle$ is maximum 
in $P$. We denote the 
lexicographically smallest input 
point that is extremal for $v$ 
by $e_v$.  The canonical directions 
are described in the following lemma, 
whose proof we postpone to 
Section~\ref{sec:canonicalDir}. They 
are computed in the learning phase. 
(Refer to Definition~\ref{def:cert-ch} 
and just above it for some of the 
basic notation below.)

\begin{lemma}\label{lem:canonicalDir}
  Let $k \eqdef n/\log^2 n$. There 
  is an $O(n\poly\log n)$ time procedure
  that takes $\poly(\log n)$ random inputs 
  and outputs an ordered sequence 
  $\textup{\textbf{V}} = v_1, \dots, v_k$ of 
  directions with the following
  properties (with probability at least 
  $1 - n^{-4}$ over construction). Let $P \sim \cD$.
  For $i = 1, \ldots, k$,  let $e_{i} = e_{v_i} \in P$, 
  let $X_i$ be the number of points  from $P$
  inside $\uss(e_{i}, e_{{i+1}})$, 
  and $Y_i$ the number of \emph{extremal}
  points inside $\uss(e_{i}, e_{{i+1}})$.
  Then
  \[
    \EX_{P \sim \cD}\Bigl[\sum_{i=1}^k X_i \log(Y_i + 1)\Bigr] = 
    O(n \log\log n).
  \]
\end{lemma}

We construct some special lines 
that are normal to the canonical directions.
The details are 
in Section~\ref{sec:ellTail}.

\begin{lemma}\label{lem:ellTail} 
  We can construct (in $O(n\poly\log n)$ time 
  with one random input) lines 
  $\ell_1, \dots, \ell_k$ with
  $\ell_i$ normal to $v_i$, and with 
  the following property (with probability 
  at least $1 - n^{-4}$ over the construction). 
  For $i = 1, \ldots, k$ (and $c$ large 
  enough), we have
  \[
    \Pr_{P \sim \cD}[|\ell_i^+ \cap P| \in [1, c\log n]] \geq 1 - n^{-3}.
  \]
\end{lemma}

We henceforth assume that the learning 
phase succeeds, so the directions
and lines have properties from
Lemma~\ref{lem:canonicalDir} and \ref{lem:ellTail}.
We call $p \in P$ is \emph{\textup{\textbf{V}}-extremal} 
if $p = e_v$ for some $v \in \textbf{V}$.
Using the canonical directions 
from Lemma~\ref{lem:canonicalDir} and the
lines from Lemma~\ref{lem:ellTail}, we 
construct a \emph{canonical hull} $\cC$ 
that is ``typical'' for random $P$.
It is the intersection of the halfplanes 
below the $\ell_i$,
i.e., $\cC = \bigcap_{i=1}^{k} \ell_i^-$.
Thus, $\cC$ is a convex polygonal region 
bounded by the $\ell_i$. 
The following corollary follows from a union bound 
of Lemma~\ref{lem:ellTail} over all $i$. 
It implies that the total number
of points outside $\cC$ is $O(n/\log n)$.

\begin{corol}\label{cor:ellTail} 
  Assume the learning phase succeeds. 
  With probability at least $1 - n^{-2}$, 
  the following holds:  for all $i$, the 
  extremal point for $v_i$ lies outside $\cC$. 
  The number of pairs $(p, s)$, where 
  $p \in P \setminus \cC$, $s$ is an edge of $\cC$,
  and $s$ is visible from $p$, is $O(n/\log n)$.
\end{corol}

\begin{figure}
  \centering
  \includegraphics{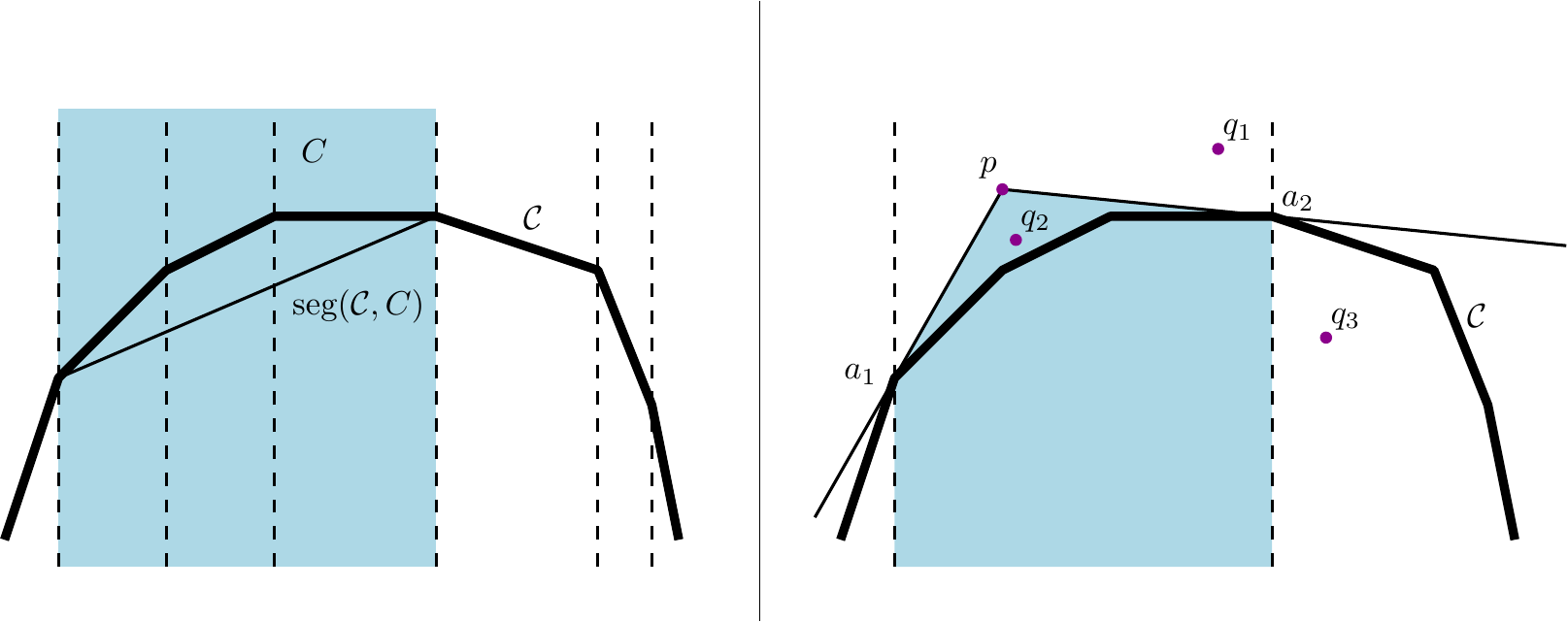} 
  \caption{(left) The $\cC$-leaf slabs
    are shown dashed. The shaded portion 
    represents a $\cC$-slab $C$. (right)
    $\pen(p)$ is shown shaded:
    $q_1$ 
    lies above the pencil: $q_2$ inside it; 
    $q_3$ is not comparable to it.}
\label{fig:pencil}
\end{figure}

To give some intuition about $\textbf{V}$, consider the simple example
where each distribution outputs a fixed point. We set $v_1$ to be the direction
pointing leftwards, so the extremal point $e_1$ is the leftmost point.
Starting from $e_1$, continue to the first extremal point $e_2$ such that
there are $O(\log n)$ extremal points between $e_1$ and $e_2$. Take any
direction $v_2$ such that $e_2$ is extremal for it. Continue in this manner
to get $\textbf{V}$. For each $v_i$, the line $\ell_i$ is normal to $v_i$ and has $\Theta(\log n)$ points
above it. So $\cC = \bigcap_{i=1}^{k} \ell_i^-$ is ``well under" $\UH(P)$, but not too far
below.

We list some preliminary concepts related to $\cC$,
see Fig.~\ref{fig:pencil}.
By drawing a vertical line through each 
vertex of $\cC$, we obtain a subdivision of 
the plane into vertical open slabs, the 
\emph{$\cC$-leaf-slabs}. A contiguous
interval of $\cC$-leaf slabs is again a 
vertical slab, called \emph{$\cC$-slab}.
The $\cC$-leaf-slabs define the slab 
structure for the upper hull algorithm, 
and we use Lemma~\ref{lem:tree} to construct 
appropriate search trees $T_1, \dots, T_n$ for 
the $\cC$-leaf slabs and for each   
distribution $\cD_i$.

For a $\cC$-slab $C$, we let $\seg(\cC, C)$ 
be the line segment between the two vertices of 
$\cC$ that lie on the vertical boundaries of $C$.
Let $p$ be a point outside of $\cC$, and let $a_1$ 
and $a_2$ be the vertices of $\cC$ where the 
two tangents for $\cC$ through $p$ touch 
$\cC$. The \emph{pencil slab} for $p$ is the
$\cC$-slab bounded by the vertical lines through 
$a_1$ and $a_2$.  The \emph{pencil of $p$}, 
$\pen(p)$ is the region inside the pencil slab for 
$p$ that lies below the line segments $\overline{a_1p}$ 
and $\overline{pa_2}$.  A point $q$ is 
\emph{comparable} to $\pen(p)$ if it lies 
inside the pencil slab for $p$. It lies 
\emph{above} $\pen(p)$ if it is comparable to
$\pen(p)$ but not inside it.

\subsection{Restricted Convex Hull Certificates}\label{sec:ch-cert}

We need to refine the certificates from 
Definition~\ref{def:cert-ch}.
Recall that a upper hull certificate
has a sorted list of extremal points 
in $P$, and a witness pair for each 
non-extremal point in $P$. The points 
$(q,r)$ form a witness pair for $p$ 
if $p \in \lss(q,r)$.
A witness pair $(q,r)$ is 
\emph{extremal} if both $q$ and
$r$ are extremal; it is
\emph{\textup{$\textbf{V}$}-extremal}
if both $q$ and $r$ are 
$\textbf{V}$-extremal.  Two distinct 
extremal points $q$ and $r$ are 
called \emph{adjacent} if there is no 
extremal point with $x$-coordinate 
strictly between
the $x$-coordinates of $q$ and $r$. 
Adjacent $\textbf{V}$-extremal points
are defined analogously.

We now define a 
\emph{$\cC$-certificate} 
for $P$.  It consists of (i) a 
list of the $\textbf{V}$-extremal 
points of $P$, sorted from left 
to right; and (ii) a list that 
has a $\cC$-slab $S_p$ for every 
other point $p \in P$. 
The $\cC$-slab $S_p$ contains $p$ and 
can be of three different kinds; see
Fig.~\ref{fig:cert}. 
Either
\begin{asparaenum}
  \item $S_p$ is a $\cC$-leaf slab; or
  \item $p$ lies below $\seg(\cC, S_p)$;
    or
  \item $S_p$ is the pencil slab for a 
    $\textbf{V}$-extremal vertex $e_v$ 
    such that $p$ lies in the pencil of 
    $e_v$.
\end{asparaenum}
The following key lemma is crucial to 
the analysis.  We defer the proof to 
the next section. The reader may wish to skip
that section and proceed to learn about 
the algorithm.

\begin{figure}
  \centering
  \includegraphics{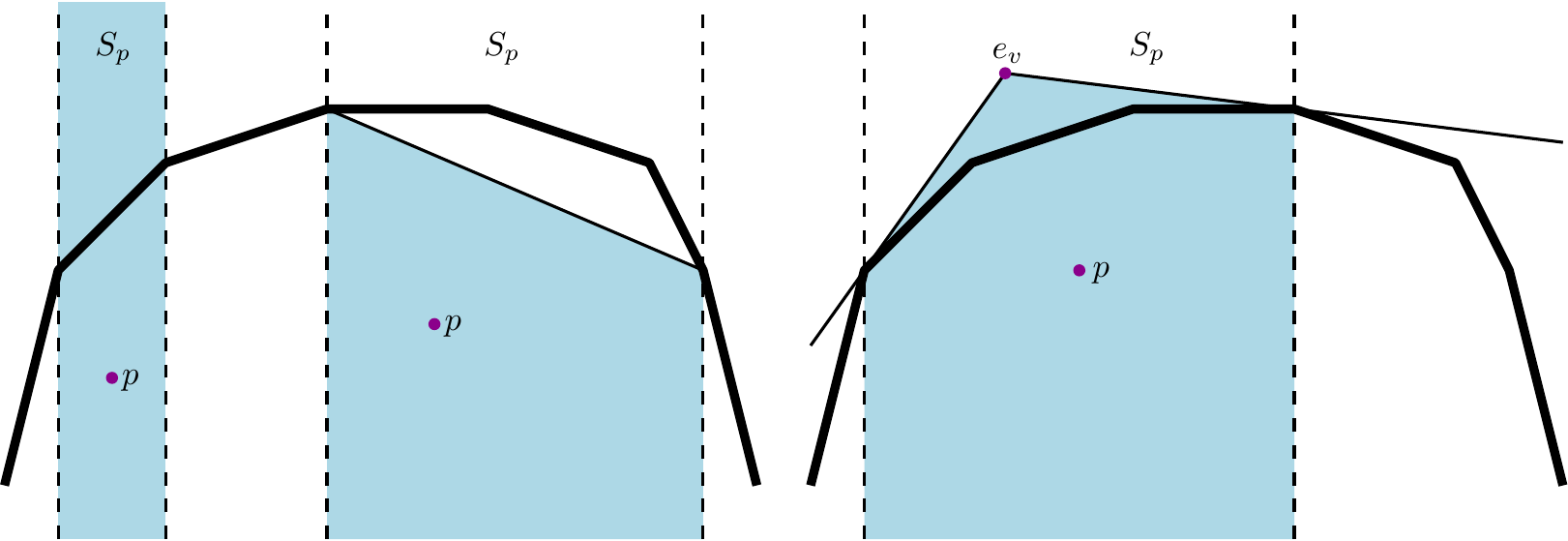} 
  \caption{The $\cC$-slab $S_p$ associated
    with $p$ can either be (i) a leaf slab; (ii) such
    that $p$ lies below $\seg(\cC, S_p)$; or (iii)
    such that $p$ lies in $\pen(e_v)$
    for a $\textbf{V}$-extremal vertex $e_v$.}
\label{fig:cert}
\end{figure}

\begin{lemma}\label{lem:entropy-sensitive-CH}
  Assume $\cC$ is obtained from a successful 
  learning phase. Let $\cT$ be 
  a linear comparison tree that computes 
  the upper hull of $P$. Then there is 
  an entropy-sensitive linear comparison tree 
  with expected depth $O(n + d_{\cT})$ that computes 
  $\cC$-certificates for $P$. 
\end{lemma}

\subsection{Proof of Lemma~\ref{lem:entropy-sensitive-CH}}
\label{sec:certificate}

The proof goes through several 
intermediate steps that successively 
transform a upper hull certificate 
into a $\cC$-certificate. 
Each step incurs expected linear 
overhead. Then, it suffices to 
apply Lemma~\ref{lem:lin->entropy} to obtain an
entropy-sensitive comparison tree with 
the claimed depth.
A certificate $\gamma$ is \emph{extremal} 
if all witness pairs in 
$\gamma$ are extremal. We provide the
required chain 
of lemmas and
give each proof in a different subsection.
The following lemma is proved in 
Section~\ref{sec:regular->extremal}.

\begin{lemma}\label{lem:regular->extremal}
  Let $\cT$ be a linear comparison tree 
  for $\UH(P)$. There exists a linear 
  comparison tree with expected depth 
  $d_\cT + O(n)$ that computes an 
  extremal certificate for $P$.
\end{lemma}

A certificate is 
\emph{\textup{$\textbf{V}$}-extremal} if 
it contains (i) a list of the 
$\textbf{V}$-extremal points of $P$, sorted 
from left to right; and (ii) a list that stores 
for every other point $p \in P$ \emph{either} 
a $\textbf{V}$-extremal witness pair
for $p$ \emph{or} two adjacent 
$\textbf{V}$-extremal points $e_1$ and $e_2$ such 
that $x(e_1) \leq x(p) \leq x(e_2)$. The next 
lemma is proved in 
Section~\ref{sec:extremal->V-extremal}.

\begin{lemma}\label{lem:extremal->V-extremal}
  Let $\cT$ be a linear comparison tree 
  that computes extremal certificates.
  There is a linear comparison tree 
  with expected depth $d_\cT + O(n)$ that computes 
  $\textup{\textbf{V}}$-extremal certificates.
\end{lemma}

Finally, we go from
$\textbf{V}$-extremal certificates to 
$\cC$-certificates. 
The proof is in Section~\ref{sec:V-extremal->canonical}.

\begin{lemma}\label{lem:V-extremal->canonical}
  Let $\cT$ be a linear comparison tree 
  that computes $\textup{\textbf{V}}$-extremal 
  certificates. There is a linear 
  comparison tree with expected depth 
  $d_{\cT} + O(n)$ that computes $\cC$-certificates. 
\end{lemma}

Lemma~\ref{lem:entropy-sensitive-CH} follows by
combining Lemmas~\ref{lem:regular->extremal},
\ref{lem:extremal->V-extremal} 
and~\ref{lem:V-extremal->canonical}
with Lemma~\ref{lem:lin->entropy}.

\subsubsection{Proof of Lemma~\ref{lem:regular->extremal}}
\label{sec:regular->extremal}
 
\begin{figure}
  \centering
  \includegraphics[scale=0.6]{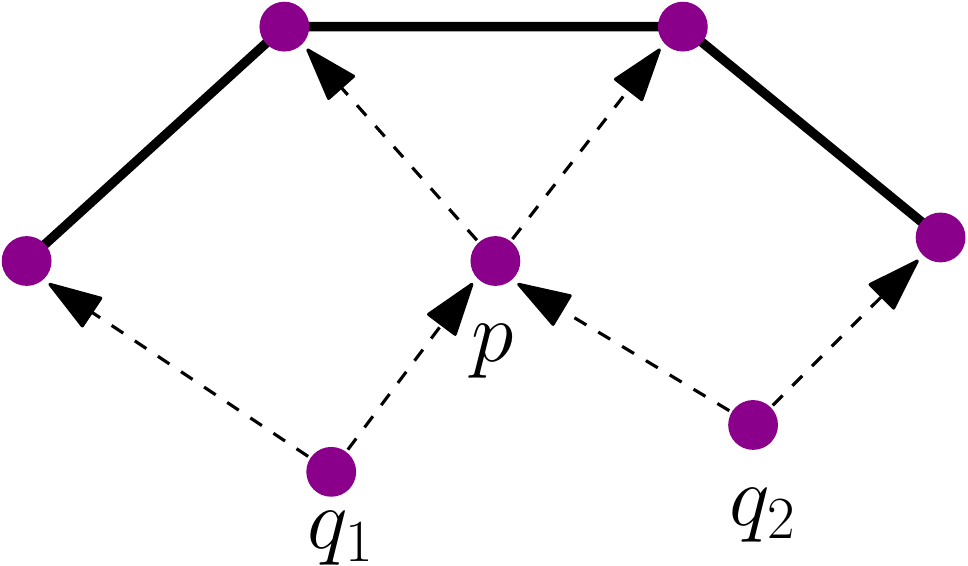}
  \caption{ The \textbf{shortcut} operation: 
    observe that computing the upper hull 
    of the out-neighbors of $p, q_1$, and $q_2$ suffices
    for removing $p$ from all witness pairs.}
  \label{fig:shortcut}
\end{figure}

We transform $\cT$ into a tree for 
extremal certificates. Since each leaf 
$v$ of $\cT$ corresponds to a certificate 
that is valid for all $P \in \cR_v$, it suffices 
to show how to convert a given certificate 
$\gamma$ for $P$ to an extremal 
certificate by performing $O(n)$ additional 
comparisons on $P$. We 
describe an algorithm for this task.

The algorithm uses two data structures: 
(i) a directed graph $G$ whose vertices 
are a subset of $P$; and (ii) a stack $S$.  
Initially, $S$ is empty and $G$ has 
a vertex for every $p \in P$. For each 
non-extremal $p \in P$, we add 
two directed edges $pq$ and $pr$ to $G$, 
where $(q,r)$ is the witness pair 
for $p$ according to $\gamma$. 
In each step, the algorithm performs 
one of the following operations, until 
$G$ has no more edges left (we will 
use the terms \emph{point} and 
\emph{vertex} interchangeably, since 
we always mean some $p \in P$).
\begin{itemize}
  \item \textbf{Prune.} If $G$ has 
    a non-extremal vertex $p$ with 
    indegree zero, we delete $p$ from 
    $G$ (together with its outgoing
    edges) and push it onto $S$.
	
  \item \textbf{Shortcut.} If $G$ 
    has a non-extremal vertex $p$ with
     indegree $1$ or $2$, we find 
     for each in-neighbor $q$ of $p$ 
     a witness pair that does not include 
     $p$, and we replace the out-edges 
     from $q$ by edges to this 
     new pair. (We explain shortly 
     how to do this.) The indegree of 
     $p$ is now zero. 
\end{itemize}
An easy induction shows that the 
algorithm maintains the following 
invariants: (i) all non-extremal 
vertices in $G$ have out-degree $2$; 
(ii) all extremal vertices of $G$ 
have out-degree $0$; (iii) for 
each non-extremal vertex $p$ of $G$, 
the two out-neighbors of $p$ 
constitute a witness pair for $p$; 
(iv) every $p \in P$ is either in $G$ 
or in $S$, but never both; (iv) when 
a point $p$ is added to $S$, then 
we have a witness pair $(q,r)$ for 
$p$ such that $q,r \notin S$. 

We analyze the number of comparisons 
on $P$. \textbf{Prune} needs no 
comparisons.  \textbf{Shortcut} is 
done as follows: we consider for each
in-neighbor $q$ of $p$ the upper convex 
hull $U$ for $p$'s two out-neighbors
and $q$'s other out-neighbor, and we 
find the edge of $U$ that lies above 
$q$. Since the $U$ constant size and 
since $p$ has in-degree at most $2$, 
this takes $O(1)$ comparisons,
see Fig.~\ref{fig:shortcut}.
There are at most $n$ calls to
\textbf{Shortcut}, so the total number 
of comparisons is $O(n)$.  Deciding which 
operation to perform depends solely on 
$G$ and requires no comparisons on $P$. 

We now argue that the algorithm cannot get 
stuck. That means that if $G$ has at 
least one edge, \textbf{Prune} or 
\textbf{Shortcut} can be applied. 
Suppose that we cannot perform 
\textbf{Prune}. Then each non-extremal 
vertex has in-degree at least $1$. 
Consider the subgraph $G'$ of $G$ 
induced by the non-extremal vertices. 
Since all extremal vertices have out-degree $0$,
all vertices in $G'$ have in-degree 
at least $1$. The average out-degree 
in $G'$ is at most $2$, so there must be 
a vertex with in-degree (in $G'$) $1$ or 
$2$.  This in-degree is the same in $G$, 
so \textbf{Shortcut} can be applied.

Thus, we can perform \textbf{Prune} or 
\textbf{Shortcut} until $G$ has no more 
edges and all non-extremal points are on 
the stack $S$.  Now we pop the points from 
$S$ and find extremal witness pairs for them.
Let $p$ be the next point on $S$.  By invariant 
(iv), there is a witness pair $(q,r)$
for $p$ whose vertices are not on $S$. Thus, 
each $q$ and $r$ is either extremal or we 
have an extremal witness pair for it.
Therefore, we can find an extremal witness 
pair for $p$ with $O(1)$ comparisons, as 
in \textbf{Shortcut}. We repeat this process 
until $S$ is empty. This takes $O(n)$ comparisons
overall, so we obtain an extremal 
certificate $\gamma'$ from $\gamma$ with $O(n)$ 
comparisons on $P$. 

\subsubsection{Proof of Lemma~\ref{lem:extremal->V-extremal}}
\label{sec:extremal->V-extremal}

\begin{figure}
  \centering
  \includegraphics{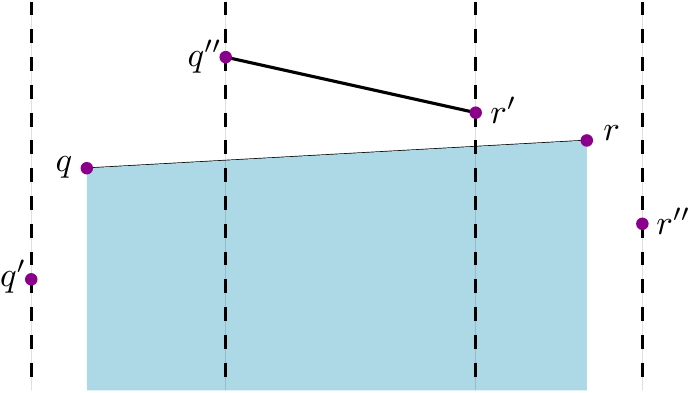}
  \caption{If $p$ is in the blue region then 
    $x(p) \in [x(q'),x(q'')]$, $p \in \lss(q'',r')$, 
    or $x(p) \in [x(r'),x(r'')]$.}
  \label{fig:v-extremal}
\end{figure}

As in Section~\ref{sec:regular->extremal}, it 
suffices to show how to convert a given 
extremal certificate into a 
$\textbf{V}$-extremal one
with $O(n)$ comparisons on $P$. This 
is done as follows.  First, we determine 
the $\textbf{V}$-extremal points on 
$\UH(P)$. This takes $O(n)$
comparisons by a simultaneous traversal 
of $\UH(P)$ and $\textbf{V}$.
Without further comparisons, we can 
now find for each extremal
point $p$ in $P$ the two adjacent 
$\textbf{V}$-extremal points that
have $p$ between them. This information 
is stored in the $\textbf{V}$-extremal
certificate.

Now let $p \in P$ be non-extremal, and 
let $(q, r)$ be the corresponding
extremal witness pair. We show how 
to find either a $\textbf{V}$-extremal 
witness pair or the right pair of 
adjacent $\textbf{V}$-extremal points.
We have determined adjacent 
$\textbf{V}$-extremal points $q', q''$ 
such that $x(q) \in [x(q'), x(q'')]$. 
(If $q$ is itself $\textbf{V}$-extremal, 
set $q' = q'' = q$.) Similarly, define 
adjacent $\textbf{V}$-extremal points 
$r', r''$. We know that $p$ lies in
$\lss(q,r)$ and hence 
$x(p) \in [x(q),x(r)]$.  Furthermore, 
the points $q',q,q'',r',r,r''$ are 
in convex position.  Since $p$ is in 
$\lss(q,r)$, one of the following 
must happen: $x(p) \in [x(q'),x(q'')]$, 
$p$ lies in $\lss(q'',r')$, or 
$x(p) \in [x(r'),x(r'')]$; see 
Fig.~\ref{fig:v-extremal}.  
We can determine which in 
$O(1)$ comparisons.

\subsubsection{Proof of Lemma~\ref{lem:V-extremal->canonical}}
\label{sec:V-extremal->canonical}

As in Sections~\ref{sec:regular->extremal} 
and~\ref{sec:extremal->V-extremal}, we convert a
$\textbf{V}$-extremal certificate $\gamma$
into a $\cC$-certificate with $O(n)$ 
expected comparisons.
For each $\textbf{V}$-extremal point in 
$\gamma$, we perform a binary search to 
find the $\cC$-leaf slab that contains 
it. This requires $o(n)$ comparisons, 
since there are at most $n/\log^2 n$ 
$\textbf{V}$-extremal points and 
since each binary search needs 
$O(\log n)$ comparisons. Next, we 
check for each $i \leq k$ if 
the extremal point for $v_i$ lies 
in $\ell^+_i$. This takes one 
comparison per point. If any check 
fails, we declare failure and use 
binary search to find for every 
$p \in P$ a $\cC$-leaf slab that 
contains it.

We now assume that there exists a 
$\textbf{V}$-extremal point
in every $\ell^+_i$. (This implies 
that all $\textbf{V}$-extremal points
lie outside $\cC$.) We use binary search 
to determine the pencil of each 
$\textbf{V}$-extremal point. Again, 
this takes $o(n)$ comparisons.
Now let $p \in P$ be not 
$\textbf{V}$-extremal. We use $O(1)$ 
comparisons and either find the slab $S_p$
or determine that $p$ lies above $\cC$.
The certificate $\gamma$ assigns to 
$p$ two $\textbf{V}$-extremal points 
$e_1$ and $e_2$ such that either 
(i) $(e_1, e_2)$ is a 
$\textbf{V}$-extremal witness pair 
for $p$; or (ii) $e_1$ and $e_2$ 
are adjacent and 
$x(e_1) \leq x(p) \leq x(e_2)$. 
We define $f_1$ as the rightmost
visible point of $\cC$ from $e_1$ 
and $f_2$ as the leftmost visible
point from $e_2$. 

Let us consider the first case; see 
Fig.~\ref{fig:canonical}(left). The
point $p$ is below $\overline{e_1 e_2}$. 
Since $e_1, f_1, f_2, e_2$ are in 
convex position, $\overline{e_1 e_2}$ is 
below their upper hull.
This means that one of the following 
holds:
$x(p) \in [x(e_1),x(f_1)]$, 
$x(p) \in [x(f_2),x(e_2)]$, or $p$
is below $\overline{f_1 f_2}$. 
This can be determined in $O(1)$
comparisons. In the first two 
cases, $p$ lies in a pencil (and 
hence we find an appropriate $S_p$),
and in the last case, we find a 
witness $\cC$-slab. Now for the second case. 
We need the following claim.

\begin{claim}\label{clm:overlap}
  If, for all $i$ there 
  is a $\textbf{V}$-extremal 
  point in $\ell^+_i$, then the 
  pencils of any two adjacent 
  $\textbf{V}$-extremal points 
  either overlap or share a slab boundary.
\end{claim}

\begin{proof} 
Refer again to 
Fig.~\ref{fig:canonical}(left). 
Let $e_1$ and $e_2$ be two adjacent 
$\textbf{V}$-extremal vertices 
such that their pencil slabs neither 
overlap nor share a boundary. Then 
$f_1$ is not visible from $e_2$. 
Consider the edge $a$ of $\cC$ 
where $f_1$ is the left endpoint. 
The edge $a$ is not visible from 
either $e_1$ or $e_2$ and is 
between them. By assumption, 
there is an extremal point 
$x$ of $P$ that sees $a$. But 
the point $x$ cannot lie to the 
left of $e_1$ or to the right 
of $e_2$ (that would violate the 
extremal nature of $e_1$ or $e_2$). 
Hence, $x$ must be between $e_1$ and 
$e_2$, contradicting the fact that 
they are adjacent.
\end{proof}

Claim~\ref{clm:overlap} implies that
$p$ is comparable to one of 
$\pen(e_1)$, $\pen(e_2)$.
By $O(1)$ comparisons, we can check 
if $p$ is contained in either 
pencil or is above $\cC$.
\begin{figure}
  \centering
  \includegraphics{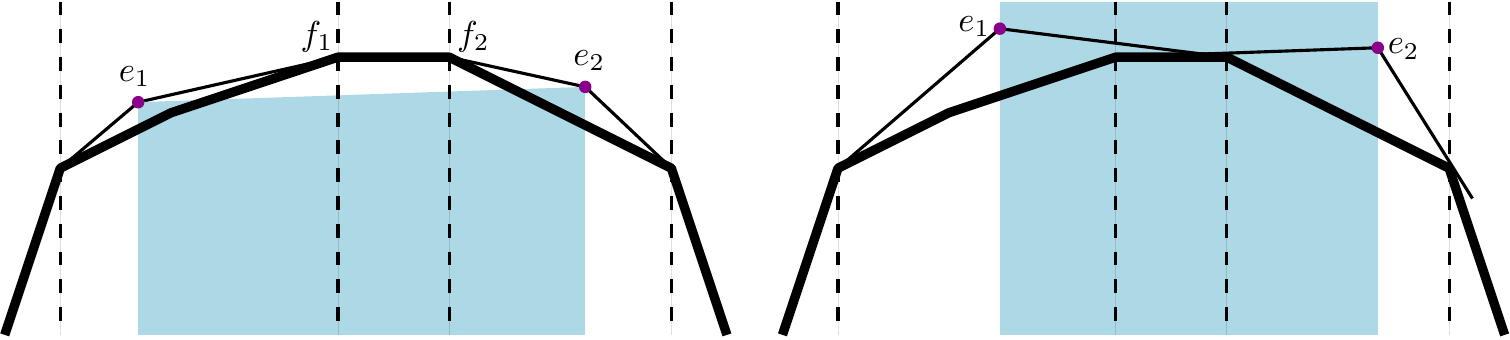}
  \caption{$\cC$-certificates: in each part, 
   $p$ is contained in the shaded region.}
  \label{fig:canonical}
\end{figure}
Finally, for all points determined 
to be above $\cC$, we use binary search
to place them in a $\cC$-leaf slab. This 
gives an appropriate $S_p$ for each 
$p \in P$, and the canonical certificate 
is complete.
We analyze the total number of comparisons.
Let $X$ be the indicator random variable 
for the event that there exist some 
$\ell^+_i$ without a $\textbf{V}$-extremal 
point. Let $Y$ denote the number of 
points above $\cC$. By 
Corollary~\ref{lem:ellTail}, $\EX[X] \leq n^{-3}$
and $\EX[Y] = O(n/\log n)$. 
The number of comparisons is at most 
$O(Xn\log n + n + Y\log n)$,
the expectation of which is $O(n)$.

\subsection{The algorithm}

Finally, we are ready to 
describe the details of our convex 
hull algorithm. It has two parts: 
the \emph{location algorithm} and the 
\emph{construction algorithm}. The 
former algorithm determines the location 
of the input points with respect to the 
canonical hull $\cC$. It must be careful 
to learn just the right amount of information 
about each point. The latter algorithm uses 
this information to compute the convex
hull of $P$ quickly.

\subsubsection{The location algorithm}
\label{sec:loc_alg}
Using Lemma~\ref{lem:tree},
we obtain near-optimal search trees 
$T_i$ for the $\cC$-leaf slabs. 
The algorithm searches progressively 
for each $p_i \in P$ in its 
tree $T_i$. Again, we interleave the
coordinate searches, and we abort the
search for a point as soon as we have 
gained enough information about
it. The location algorithm maintains 
the following information. 

\begin{itemize}
\item 
\textbf{Current slabs $C_i$.}
  For each point $p_i \in P$, 
  we store a current $\cC$-slab $C_i$ 
  containing $p_i$ that corresponds 
  to a node of $T_i$. 
\item \textbf{Active points $A$.}
  The active points are stored in a 
  priority-queue $L(A)$ as in 
  Claim~\ref{clm:ds}. The key associated 
  with an active point $p_i \in A$ is the size 
  of the associated current slab $C_i$ 
  (represented as an integer between $1$ and $k$).
\item \textbf{Extremal candidates $\tilde{e}_v$.}
  For each canonical direction 
  $v \in \textbf{V}$, we store a point
  $\tilde{e}_v \in P$ that lies outside 
  of $\cC$.  We call $\tilde{e}_v$ an 
  \emph{extremal candidate} for $v$.
\item
  \textbf{Pencils for the points outside of $\cC$.}
  For each point $p$ that has been 
  located outside of $\cC$, we store its
  pencil $\pen(p)$.
\item
  \textbf{Points with the left- and rightmost pencils.}
  For each edge $s$ of $\cC$, we store two  
  points $p_{s1}$
  and $p_{s2}$ such that 
  (i) $p_{s1}$ and $p_{s2}$ lie outside of $\cC$;  
  (ii) $s$ lies in $\pen(p_{s1})$ and $\pen(p_{s2})$;
  (iii) among all pencils seen 
  so far that contain $s$, the left boundary of 
  $\pen(p_{s1})$  lies furthest to the left 
  and the right boundary of $\pen(p_{s2})$  
  lies furthest to the right. 
\end{itemize}

\begin{figure}
  \centering
  \includegraphics{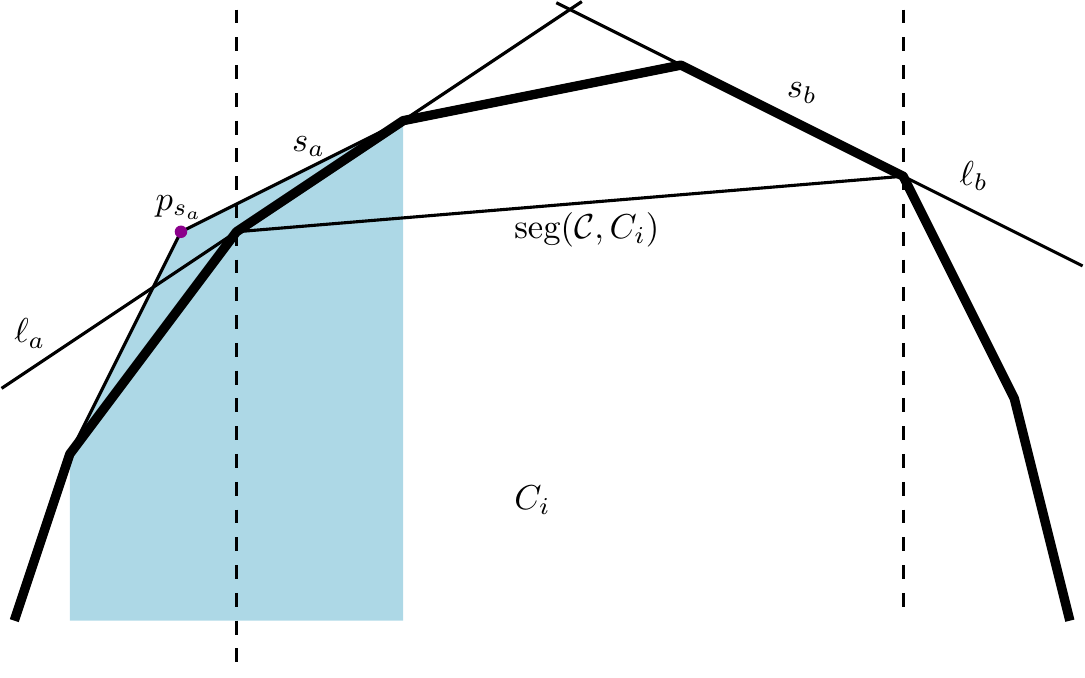}
  \caption{The algorithm:  the boundary 
  of $C_i$ is shown dashed, the pencil 
  $\pen(p_{s_a})$ is shaded.}
  \label{fig:algorithm}
\end{figure}

Initially, we set $A = P$ and each 
$C_i$ to the root of the corresponding search
tree $T_i$. The extremal candidates 
$\tilde{e}_v$ as well as the points 
$p_{s1}, p_{s2}$ with the left- and 
rightmost pencils are set to the 
null pointer. The location algorithm 
proceeds in \emph{rounds}. In each round, 
we perform a \findmax{} on $L(A)$. 
Suppose that \findmax{} returns $p_i$.
We compare $p_i$ with the vertical 
line that corresponds to its current node
in $T_i$ and advance $C_i$ to 
the appropriate child. This reduces the
size of $C_i$, so we also perform a 
\deckey{} on $L(A)$. Next, we distinguish 
three cases:

\noindent\textbf{Case 1}: $p_i$ lies below 
$\seg(\cC, C_i)$. We declare $p_i$ 
inactive and \delete{} it from $L(A)$.

For the next two cases, we know 
that $p_i$ lies above $\seg(\cC, C_i)$.
Let $\ell_a$, $\ell_b$ be the canonical 
lines that support the edges
$s_a$ and $s_b$ of $\cC$ that are 
incident to the boundary vertices of
$C_i$ and lie inside of $C_i$; see 
Fig.~\ref{fig:algorithm}.
We check where $p_i$ lies with respect 
to $\ell_a$ and $\ell_b$.

\noindent\textbf{Case 2}:  
$p_i$ is above $\ell_a$ or above $\ell_b$. 
This means that $p_i$ is outside of $\cC$.  
We declare $p_i$ inactive and
\delete{} it from $L(A)$. Next, we 
perform a binary search to find 
$\pen(p_i)$ and all the edges of 
$\cC$ that are visible 
from $p_i$. For each such edge $s$,
we compare $p_i$ with the extremal 
candidate for $s$, and if $p_i$ is more 
extreme in the corresponding direction, 
we update the extremal candidate accordingly. 
We also update the points $p_{s1}$ and 
$p_{s2}$ to $p_i$, if necessary. 

\noindent\textbf{Case 3}: $p_i$ lies 
below $\ell_a$ and $\ell_b$. Recall that
$\ell_a$ corresponds to the edge $s_a$ of 
$\cC$ and $\ell_b$ corresponds to the 
edge $s_b$ of $\cC$. We take the rightmost 
pencil for $s_a$ and the leftmost pencil for $s_b$ 
(if they exist); see Fig.~\ref{fig:algorithm}.
We compare $p_i$ with these pencils. 
If $p_i$ lies inside a pencil, we are done.
If $p_i$ is above a pencil, we learn 
that $p_i$ lies outside of $\cC$, and we 
process as in Case~2. In both situations, 
we declare $p_i$ inactive and \delete{} it 
from $L(A)$. 
If neither of these happen, $p_i$ remains active.

The location algorithm continues until $A$ is 
empty (note that every point becomes inactive 
eventually, because as soon as $C_i$ is a leaf
slab, either Case 1 or Case 2 applies).

\subsubsection{Running time of the location algorithm}

We now analyze the running time of 
the location algorithm, starting
with some preliminary claims.
The algorithm is deterministic, so 
we can talk of deterministic 
properties of the behavior on any input.

\begin{claim}\label{clm:algoext} 
  Fix an input $P$. Let
  $e_v \in P$ be $\textbf{V}$-extremal,
  and let $S$ be the pencil
  slab for $e_v$. Once the 
  search for $e_v$ reaches a slab $D$ 
  with $|D| \leq |S|$, 
  $e_v$ will be identified as an extremal point 
  for direction $v$.
\end{claim}

\begin{proof} 
At least one vertical boundary 
line of $D$ lies inside (the 
closure of) $S$ and 
$D \cap S$ contains at least 
one leaf slab. Since $S$ is a 
pencil slab, $e_v$ sees all edges 
of $\cC$ in $D \cap S$, so one of 
the boundary edges $s_a$ or $s_b$ 
corresponding to $D$, as used in 
Cases 2 and 3 of the algorithm 
(see Fig.~\ref{fig:algorithm}), must
be visible to $e_v$. 
Hence, $e_v$ lies in 
$\ell_a^+ \cup \ell_b^+$, and
this is detected in Case~2 
of the location algorithm.
\end{proof}

\begin{figure}
  \centering
  \includegraphics{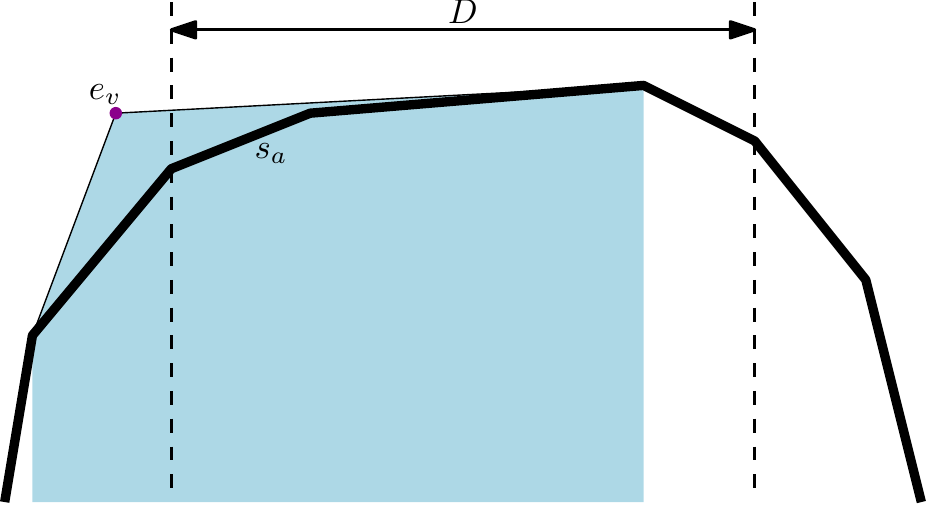}
  \caption{The left boundary of slab $D$ is 
    contained in the pencil slab of $e_v$.}
  \label{fig:case-3}
\end{figure}

\begin{claim}\label{clm:algo-inact} 
  Let $e_v \in P$ be 
  $\textbf{V}$-extremal, and 
  $S$ the pencil slab for $e_v$.
  Suppose $p \in P$ lies in 
  $\pen(e_v)$. Once
  the search for $p$ reaches 
  a slab $D$ with $|D| \leq |S|$,
  the point $p$ becomes inactive  
  in the next round that it is 
  processed.
\end{claim}

\begin{proof} 
Consider the situation after 
the round in which $p$ reaches $D$ 
with $|D| \leq |S|$. 
The location algorithm schedules 
points according to the size 
of their current slab. 
Thus, when $p$ is processed next, 
all other active points are placed 
in slabs of size at most $|S|$. 
By Claim~\ref{clm:algoext},
if $e_v$ is ever placed in slab of size 
at most $|S|$, the algorithm
detects that it is $\textbf{V}$-extremal 
and makes it inactive.

Hence, when $p$ is processed next, 
$e_v$ has been identified
as the extremal point in direction
$v$. Note that $D \cap S \neq \emptyset$,
since $p \in D \cap S$. Some 
boundary (suppose it is the left one) of $D$
lies inside $S$. Let $s_a$ be the 
corresponding edge of $\cC$, as used by 
the location algorithm; see Fig.~\ref{fig:case-3}.
Since $s_a$ is visible from $e_v$, 
and since $e_v$ has been processed,
it follows that the pencil 
slab of the rightmost pencil 
for $s_a$ spans all of $D \cap S$.
In Case~3 of the location algorithm 
(in this round), $p$
will either be found inside this 
pencil, or outside of $\cC$.
Either way, $p$ becomes inactive.
\end{proof}

We arrive at the main lemma of this section.

\begin{lemma}\label{lem:algoCH} 
  The total number of rounds in 
  the location algorithm is 
  $O(n + \textup{\OPTCH})$.
\end{lemma}

\begin{proof} 
Let $\cT$ be an entropy-sensitive 
comparison tree that 
computes a $\cC$-certificate for 
$P$ in expected depth $O(n + \OPTCH)$. 
Such a tree exists by 
Lemma~\ref{lem:entropy-sensitive-CH}.
Let $v$ be a leaf of $\cT$ with 
$d_v \leq n^2$. By 
Proposition~\ref{prop:Cartesian},
$\cR_v$ is a Cartesian product 
$\cR_v = \prod_{i=1}^n R_i$.
The depth of $v$ is 
$d_v = -\sum_{i=1}^n \log \Pr[p_i \in R_i]$, 
by Proposition~\ref{prop:entropyDepth}. 
Now consider a random input $P$, 
conditioned on $P \in \cR_v$. We 
show that expected number of rounds 
for $P$ is  $O(n + d_v)$. This
also holds for $d_v > n^2$, since
there are never more than $n^2$ rounds.
The lemma follows, as the 
expected number of rounds is
\[ 
  \sum_{v \text{ leaf of } \cT} \Pr[P \in \mathcal{R}_v] O(n + d_v) =
  O(n + d_\cT).
\]
Let $v$ be a leaf with $d_v \leq n^2$ and $\gamma$ the 
$\cC$-certificate for 
$v$. The main technical argument is 
summarized in the following claim.

\begin{claim}\label{clm:algoCH} 
  Let $P \in \cR_v$ and
  $p_i \in P$. The number 
  of rounds involving 
  $p_i$ is at most one more than 
  the number of steps required for an
  $S_{p_i}$-restricted search for 
  $p_i$ in $T_i$.
\end{claim}

\begin{proof} 
By definition of $\cC$-certificates, 
$S_{p_i}$ is one of three types. 
Either $S_{p_i}$ is a $\cC$-leaf slab, 
$p_i$ is below $\seg(S_{p_i},\cC)$,
or $S_{p_i}$ is a pencil slab of a 
$\textbf{V}$-extremal vertex. 
In all cases, $S_{p_i}$ contains $R_i$.
When $S_{p_i}$ is a leaf slab, an
$S_{p_i}$-restricted search for $p$ is 
a complete search. Hence, this is 
always at least the number of rounds 
involving $p_i$.
Suppose $p_i$ is below $\seg(S_{p_i},\cC)$. 
For any slab $S \subseteq S_{p_i}$,
$\seg(S,\cC)$ is above $\seg(S_{p_i},\cC)$. 
If $p_i$ is located in any slab $S \subseteq S_{p_i}$,
it is made inactive (Case 1 of the algorithm). 

Now for the last case. 
The slab $S_{p_i}$ is the pencil slab 
for a $\textbf{V}$-extremal vertex $e_v$,
such that the $\pen(e_v)$, contains $p_i$. 
Suppose the search for $p_i$ leads to 
slab $D \subseteq S_{p_i}$ and $p_i$
is still active.
By Claim~\ref{clm:algo-inact}, since 
$|D| \leq |S_{p_i}|$, $p_i$ becomes
inactive in the next round.
\end{proof}

Suppose $P$ is chosen randomly from 
$\cR_v$. The distribution restricted 
to $p_i$ is simply random from $R_i$.  
By Lemma~\ref{lem:search-time}, the 
expected $S_{p_i}$-restricted search 
time is $O(1-\log \Pr[p \in R_i])$. Combining 
with Claim~\ref{clm:algoCH}, the expected
number of rounds is 
\[
  O(n -\sum_{i=1}^n \log \Pr[p_i  \in R_i]) = O(n + d_v).
\]
\end{proof}

\begin{lemma} 
  The expected running time of 
  the location algorithm is 
  $O(n + \textup{\OPTCH})$.
\end{lemma}

\begin{proof}
By Claim~\ref{clm:ds},
the total overhead for the heap 
structure is linear in the number
of rounds. The time to implement 
Cases~1 and~3 is $O(1)$, as we only 
need to compare $p_i$ with a 
constant number of lines. Hence, the 
total time for this is at most proportional to the 
number of rounds. 

In Case~2, we do a binary search 
for $p_i$ and possibly update 
an extremal point (and pencil) 
for each edge 
visible from $p_i$. The case
only occurs if $p_i$ lies 
outside $\cC$.  By 
Corollary~\ref{cor:ellTail}, the 
expected number such updates is $O(n/\log n)$. 
Overall, the total cost for Case~2 
operations is $O(n)$. Combining with 
Lemma~\ref{lem:algoCH}, the expected 
running time is 
$O(n + \textup{\OPTCH})$.
\end{proof}

\subsubsection{The construction algorithm}

We now describe the upper hull 
construction that uses the information
from the location algorithm to compute
$\UH(P)$ quickly. First, we dive 
into the geometry of pencils.

\begin{claim}\label{clm:no_pencil_overlap}
  Suppose that all 
  \textup{$\textbf{V}$}-extremal 
  points of $P$ lie outside
  of $\cC$, and let $e_v$ be a 
  \textup{$\textbf{V}$}-extremal point. 
  Then $e_v$ does not lie
  in the pencil of any other point  
  outside $\cC$.
\end{claim}

\begin{proof}
Suppose that  $e_v \in \pen(p)$ for 
another point $p \in P$ outside of
$\cC$.  Then a vertex of $\pen(p)$ would
be more extremal in direction $v$ than 
$e_v$. It cannot be $p$, since then
$e_v$ would not be extremal in direction 
$v$. It also cannot be a
vertex of $\cC$, because $e_v$ lies in 
$\ell_v^+$, while all vertices of $\cC$
lie on $\ell_v$ or in $\ell_v^-$. 
Thus, $p$ cannot exist.
\end{proof}

\begin{figure}
  \centering
  \includegraphics{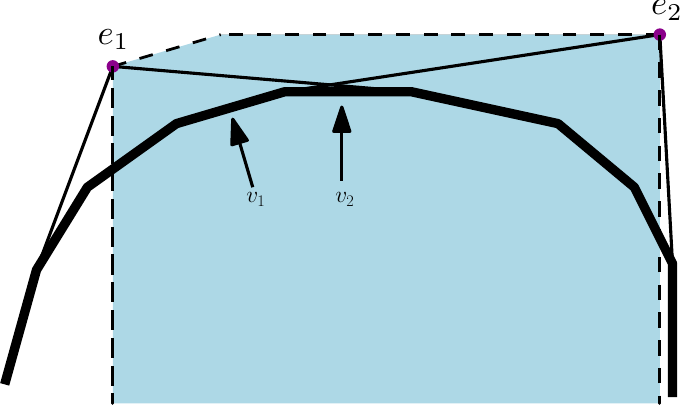}
  \caption{The lines perpendicular 
    to directions $v_1$ and $v_2$ define 
  the upper boundary of the shaded region 
  where $p$ lies. All edges 
  seen by a point in the shaded region 
  can be seen by either $e_1$ or $e_2$.}
  \label{fig:union}
\end{figure}

\begin{claim}\label{clm:extremal_pencil}
  Suppose \textup{$\textbf{V}$}-extremal 
  points of $P$ lie outside
  of $\cC$. Let $e_1$ and $e_2$ be 
  two adjacent \textup{$\textbf{V}$}-extremal 
  points and let $p \in P$ be above $\cC$ 
  such that the $x$-coordinate of
  $p$ lies between the $x$-coordinates of $e_1$ 
  and $e_2$.  Then, the portion of $\pen(p)$ 
  below $\cC$ is contained 
  in $\pen(e_1) \cup \pen(e_2)$.
\end{claim}

\begin{proof} 
By Claim~\ref{clm:overlap}, the 
(closures of the) pencil slabs 
of $e_1$ and $e_2$ overlap.
Let $v_1$ be the last canonical 
direction for which $e_1$ is extremal 
and $v_2$ the first canonical direction 
for which $e_2$ is extremal. As $e_1$ and $e_2$ 
are adjacent, $v_1$ and $v_2$ are 
consecutive in $\textbf{V}$; see 
Fig.~\ref{fig:union}.
Consider the convex region bounded 
by the vertical downward
ray from $e_1$, the vertical 
downward ray from $e_2$, the line parallel
to $\ell_{v_1}$ through $e_1$, and the line 
parallel to $\ell_{v_2}$ through 
$e_2$. By construction, $p$ lies inside 
this convex region (the shaded
area in Fig.~\ref{fig:union}). 
By convexity, for every 
$v \in \textbf{V}$, at 
least one of $e_1$ 
or $e_2$ is more extremal with respect 
to $v$ than $p$. Hence,
any edge of $\cC$ visible from $p$ is 
visible by either $e_1$ or $e_2$.
The portion of $\pen(p)$ below 
$\cC$ is the union of 
regions below edges of $\cC$ visible 
from $p$. Therefore, it 
lies in $\pen(e_1) \cup \pen(e_2)$.
\end{proof}

As described in Section~\ref{sec:loc_alg},
the location algorithm determines 
for for each $p \in P$ that
either
(a) $p$ lies outside of $\cC$;
(b) $p$ lies inside of $\cC$, as witnessed 
by a segment $\seg(\cC, C_p)$; or
(c) $p$ lies in the pencil of a 
point located outside of $\cC$.
We also have the $\textbf{V}$-extremal 
vertices $e_v$ for all $v \in \textbf{V}$.
We now use this information in order to 
find $\UH(P)$. By Corollary~\ref{cor:ellTail}, with
probability at least $1 - n^{-2}$,
for each canonical direction in $\textbf{V}$ 
there is a extremal point
outside of $\cC$ and the total number 
of points outside $\cC$ is $O(n/\log n)$.
We assume that these conditions 
hold. (Otherwise, we can compute $\UH(P)$ 
in $O(n\log n)$ time, affecting the expected
work only by a lower order term.)

For any point $a$, the 
\emph{$\textbf{V}$-pair for $a$} is the pair
of adjacent $\textbf{V}$-extremal points 
such that $a$ lies between them.
The construction algorithm goes through a series 
of steps. The exact details
of some of these steps will be given 
in subsequent claims.
\medskip

\begin{asparaenum}
  \item \label{step-1} 
     Compute the upper hull of the 
     $\textbf{V}$-extremal points.
  \item \label{step-2} 
    For each vertex $a$ of $\cC$, compute the 
    $\textbf{V}$-pair for $a$.
  \item \label{step-3} 
    For each input 
    point $p$ outside $\cC$, compute 
    its $\textbf{V}$-pair by binary search.
  \item \label{step-4} 
    For each input point $p$ 
    below a segment $\seg(\cC,C_p)$, in 
    $O(1)$ time, either find its $\textbf{V}$-pair 
    or find a segment between $\textbf{V}$-extremal
    points above it. (Details in Claim~\ref{clm:step-4}.)
  \item \label{step-5} 
    For each input point 
    $p$ located in the pencil 
    of a non-$\textbf{V}$-extremal point, 
    in $O(1)$ time, either locate $p$ in the pencil of a 
    $\textbf{V}$-extremal point or determine 
    that it is outside $\cC$. In 
    the latter case, use binary search to find its
    $\textbf{V}$-pair.
  \item \label{step-6} 
    For each input 
    point $p$ located in the pencil 
    of an $\textbf{V}$-extremal point, in $O(1)$ 
    time, find a segment between $\textbf{V}$-extremal 
    points above it or find its 
    $\textbf{V}$-pair. (Details for 
    both steps in Claim~\ref{clm:step-56}.)
  \item \label{step-7} 
    By now, for every 
    non-$\textbf{V}$-extremal $p \in P$, 
    we found a $\textbf{V}$-pair or proved 
    $p$ non-extremal through a $\textbf{V}$-extremal 
    segment above it.  For every pair 
    $(e_1, e_2)$ of adjacent $\textbf{V}$-extremal 
    points, find the set $Q$ of points
    that lie above $\overline{e_1e_2}$. 
    Use an output-sensitive 
    upper hull algorithm~\cite{KirkpatrickSe86} to find the convex 
    hull of $Q$. Finally, concatenate the resulting 
    convex hulls to obtain $\UH(P)$. 
\end{asparaenum}

\begin{claim} \label{clm:ext-cand} 
  After the location algorithm, for each
  canonical direction $v \in V$, the extremal candidate 
  $\tilde{e}_v$ is the
  actual extremal point $e_v$ in direction $v$.  
\end{claim}

\begin{proof} 
By Claim~\ref{clm:no_pencil_overlap}, 
$e_v$ does not lie in the pencil of any other
point in $P$. Hence, the location algorithm 
classifies $e_v$ as the extremal candidate for 
$v$, and this choice does 
not change later on.
\end{proof}

\begin{claim} \label{clm:step-123} 
  The total running time for 
  Steps~\ref{step-1},\ref{step-2},\ref{step-3}
  and all binary searches in 
  Step~\ref{step-5} is $O(n)$.
\end{claim}

\begin{proof} 
There are $k = n/\log^2 n$ $\textbf{V}$-extremal 
points, so finding their upper hull takes $O(n)$ time. 
We simultaneously traverse this upper hull and $\cC$ 
to obtain the $\textbf{V}$-pairs
for all vertices of $\cC$. As there are
$O(n/\log n)$ points outside of $\cC$, the
total time for the binary searches is $O(n)$.
\end{proof}

\begin{figure}
  \centering
  \includegraphics{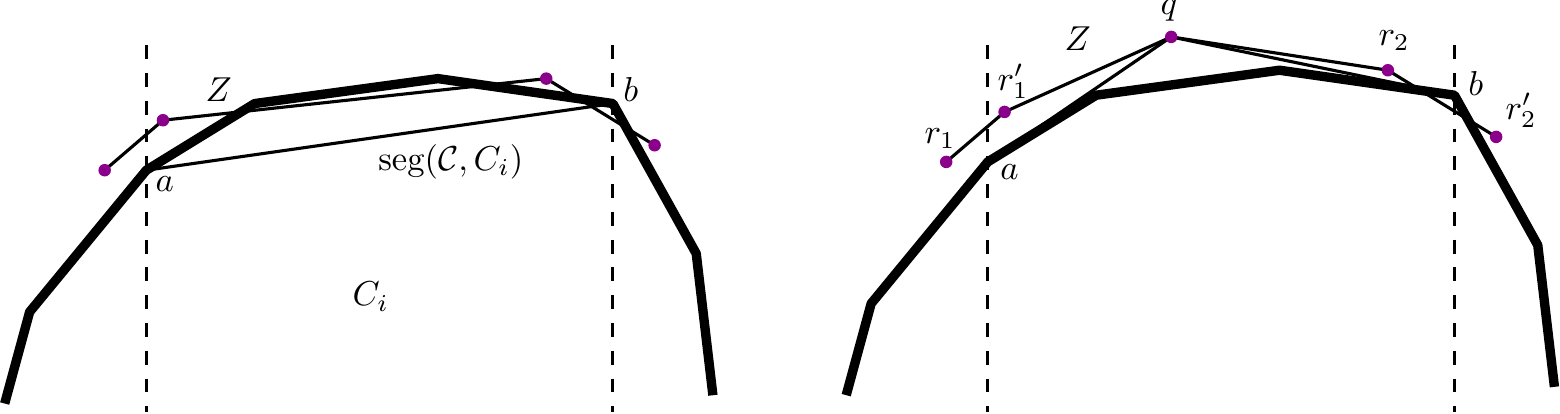}
  \caption{(left) In Step~\ref{step-4}, the 
    region below $\seg(\cC,C_i)$ is 
  either below the middle segment of $Z$
  or between one of the two $\textbf{V}$-pairs. 
  (right) In Step~\ref{step-6}, $\pen(q)$
  can be partitioned into regions below 
  $\overline{qr'_1}$, below $\overline{qr_2}$, 
  between $r_1, r'_1$, or between $r_2, r'_2$.}
  \label{fig:step-5}
\end{figure}

\begin{claim} \label{clm:step-4} 
  Suppose $p \in P$ lies below a segment $\seg(\cC, C_p)$.
  Using the information gathered before Step~\ref{step-4},
  we can either find its $\textbf{V}$-pair or a segment between 
  $\textbf{V}$-extremal points above it in $O(1)$ time.
\end{claim}

\begin{proof} 
Let $a$ and $b$ be the endpoints of 
$\seg(\cC, C_p)$. Consider the 
upper hull $Z$ of the at most four
$\textbf{V}$-extremal points  that define 
the $\textbf{V}$-pairs of $a$ and 
$b$; see Fig.~\ref{fig:step-5}(left). 
The hull $Z$ has at most three edges, 
and only the middle one (if it 
exists) might not be between two 
adjacent $\textbf{V}$-extremal points. 
If the middle edge of $Z$ exists, it lies 
strictly above $\seg(\cC, C_p)$. This is because 
the endpoints of the middle edge have $x$-coordinates 
between $x(a)$ and $x(b)$ and lie outside of 
$\cC$ (since they are $\textbf{V}$-extremal), 
while $\seg(\cC, C_p)$ is inside $\cC$.
Now we compare $p$ with the upper hull $Z$. This 
either finds a $\textbf{V}$-pair
for $p$ (if $p$ lies in the interval 
corresponding to the leftmost or rightmost 
edge of $Z$) or shows that $p$ lies below a segment between
two $\textbf{V}$-extremal points (if it lies in the
interval corresponding to the middle edge of $Z$).
\end{proof}

\begin{claim} \label{clm:step-56} 
  Suppose $p \in P$ lies in $\pen(q)$, 
  where $q$ is above $\cC$, and the 
  construction algorithm has completed 
  Step~\ref{step-4}.  If $q$ is 
  non-$\textbf{V}$-extremal, in $O(1)$ time, 
  we can either find a $\textbf{V}$-extremal point 
  $q'$ such that $p \in \pen(q')$, or determine
  that $p$ is above $\cC$. 
  If $q$ is $\textbf{V}$-extremal, then in $O(1)$ time
  we can find a $\textbf{V}$-segment above $p$ or find the 
  $\textbf{V}$-pair for $p$.
\end{claim}

\begin{proof} 
Let $q$ be non-$\textbf{V}$-extremal.  
As $q$ is outside $\cC$, we know 
the $\textbf{V}$-pair $\{e_1, e_2\}$ for $q$. 
By Claim~\ref{clm:extremal_pencil}, if
$p$ lies below $\cC$, it is in $\pen(e_1)$ or 
in $\pen(e_2)$. We can determine which (if at all)
in $O(1)$ time.

Let $q$ be $\textbf{V}$-extremal, and $a$,
$b$ the vertices of $\cC$ on the 
boundary of $\pen(q)$, where $a$ is to the
left; see Fig.~\ref{fig:step-5}(right). 
Let $(r_1, r'_1)$ be $a$'s $\textbf{V}$-pair,
where $r_1$ is to the left. 
Similarly, $(r_2, r'_2)$ is $b$'s $\textbf{V}$-pair.
The segments $\overline{qr'_1}$
and $\overline{qr_2}$ are above $\pen(q)$. Furthermore, 
the pencil slab of $q$
is between $r_1$ and $r'_2$. One of the 
following must be true for any point in $\pen(q)$:
it is below $\overline{qr'_1}$, below 
$\overline{qr_2}$, between $(r_1, r'_1)$,
or between $(r_2, r'_2)$. This can be 
determined in $O(1)$ time.
\end{proof}

We are now armed with all the facts to bound the running time.

\begin{lemma}
  With the information from the location algorithm, 
  $\UH(P)$ can be computed in expected time $O(n \log\log n)$.
\end{lemma}

\begin{proof} 
By Claims~\ref{clm:step-123},~\ref{clm:step-4}, 
and~\ref{clm:step-56}, the  
first six steps take $O(n)$ time. Let 
the $\textbf{V}$-extremal points
be ordered $e_1, \dots, e_k$. Let 
$X_i$ be the number of points
in $\uss(e_i, e_{i+1})$ and $Y_i$ 
the number of extremal points
in this set. We use an output-sensitive
upper hull algorithm, so
the running time of Step~\ref{step-7} is 
$O(\sum_{i \leq k} X_i \log (Y_i+1))$.
By Lemma~\ref{lem:canonicalDir}, this
is $O(n \log\log n)$, as desired.
\end{proof}

\section{Proofs of 
  Lemma~\ref{lem:canonicalDir} and 
  Lemma~\ref{lem:ellTail}}
\label{sec:lproofs}

We begin with some preliminaries
about projective duality and 
a probabilistic claim
about geometric constructions 
over product distributions.
Consider an input $P$.
As is well known, there is a 
\emph{dual} set $P^*$ of lines
that helps us understand 
the properties of $P$. More 
precisely, we use the standard
duality along the unit paraboloid 
that maps a point $p = (x(p), y(p))$ 
to the line $p^* : y = 2x(p)x -y(p)$
and vice versa. The \emph{lower envelope} 
of $P^*$ is the
pointwise minimum of the $n$ 
lines $p_1^*, \dots, p_n^*$, 
considered as univariate functions. 
We denote it by $\lev_0(P^*)$.
There is a one-to-one correspondence 
between the vertices and edges of 
$\lev_0(P)$ and the edges and 
vertices of $\UH(P)$. 
More generally, for $z = 0, \ldots n$, 
the \emph{$z$-level} of $P^*$ is 
the closure of the set of all points 
that lie on lines of $P^*$ and that 
have exactly $z$ lines of $P^*$ strictly 
below them.  The $z$-level is 
an $x$-monotone polygonal curve, 
and we denote it by $\lev_z(P^*)$; 
see Fig.~\ref{fig:arrangement}.
Finally, the \emph{$(\lee z)$-level} 
of $P^*$, $\lev_{\lee z}(P^*)$, is the
set of all points on lines in $P^*$ that 
are on or below $\lev_z(P^*)$. 

Consider the following abstract 
procedure. Let $b$ be a constant,
and for any set $Q$ of $b$ lines, 
let $\reg(Q)$ be some geometric region
defined by the lines in $Q$.
That is, $\reg(\cdot)$ is a function
from sets of lines of size $b$ to 
geometric regions (i.e., subsets of
of $\R^2$). For example, $\reg(\cdot)$
may be a triangle or trapezoid formed 
by some lines in $Q$. For some such 
region $R$ and a line $\ell$, let 
$\chi(\ell,R)$ be a boolean
function, taking as input a line 
and a geometric region.

Suppose we take a random 
instance $Q^* \sim \cD$, and we 
apply some procedure to determine 
various subsets $Q_1, Q_2, \ldots$ 
of $b$ lines from $Q^*$, 
chosen based on 
the sums  
$\sum_{i=1}^n \chi(q^*_i,\reg(Q_j))$.
Now generate another random 
instance $P^* \sim \cD$. What
can be say about the values of 
$\sum_i \chi(p^*_i,\reg(Q_j))$? 
We expect them to resemble 
$\sum_i \chi(q^*_i,\reg(Q_j))$, 
but we have to deal with
subtle issues of dependencies. 
In the former case, $Q_j$ actually 
depends on $Q^*$, while in the 
latter case it does not. 
Nonetheless, we can apply 
concentration inequalities
to make statements about 
$\sum_i \chi(p^*_i,\reg(Q_j))$.
Let $J \subseteq [n]$ be a set 
of $b$ indices in $[n]$, and 
set $Q^*_J := \{q^*_j \mid j \in J\}$.
The following lemma can be 
seen as generalization of Lemma~3.2 in 
Ailon \etal~\cite{AilonCCLMS11}, with
a very similar proof.

\begin{lemma}\label{lem:rand-const} 
  Let $b > 0$ an integer, and 
  $f_l(n)$, $f_u(n)$ increasing 
  functions such that 
  $f_u(n) \geq f_l(n) \geq c'b\log n$,
  for a constant $c' > 0$ and large enough
  $n$. The 
  following holds with probability at 
  least $1 - n^{-4}$ over a random 
  $Q^* \sim \cD$.  For all index sets 
  $J$ of size $b$, if 
  $\sum_{i \leq n} \chi(q^*_i,\reg(Q^*_{J})) \in [f_l(n), f_u(n)]$, 
  then for some absolute constant $\alpha \in (0,1)$,
  \[
    \Pr_{P \sim \cD}\Bigl[\sum_{i \leq n} 
      \chi(p^*_i,\reg(Q^*_{J})) \in [\alpha f_l(n), f_u(n)/\alpha]\Bigr] 
      \geq 1 - n^{-3}.
\]
\end{lemma}

\begin{proof} 
Fix an index set $J \subseteq [n]$ of size $b$,
and a set $Q_J = \{q^*_j \mid j \in J\}$
of lines. By independence, the distributions
$\cD_i$, $i \notin J$, remain unchanged, 
and we generate a random 
$Q^*$ conditioned on $Q_J$ being fixed. 
(This means that we sample random 
lines $q^*_i \sim \cD_i$, for $i \notin J$.) 
We have 
$|\sum_{i \notin J} \chi(q^*_i,\reg(Q_J)) - 
\sum_{i} \chi(q^*_i,\reg(Q_J))| \leq b$, 
so if 
$\sum_{i} \chi(q^*_i,\reg(Q_J)) \in [f_l(n), f_u(n)]$, 
then 
$\sum_{i \notin J} \chi(q^*_i,\reg(Q_J)) \in [g_l(n), g_u(n)]$,
for $g_l(n) = f_l(n) - b$ 
and $g_u(n) = f_u(n) + b$.
What can we say about 
$\sum_{i \notin J} \chi(p^*_i,\reg(Q_J))$, 
for an independent $P^* \sim \cD$? 
Since $\reg(Q_J)$ is fixed,
$\chi(p^*_i,\reg(Q_J))$ and 
$\chi(q^*_i, \reg(Q_J))$ are 
identically distributed.

Define (independent) indicator variables 
$Z_{i} = \chi(q^*_i,\reg(Q_J))$, 
and let 
$\hat{Z} = \sum_{i \notin J} Z_{i}$ and 
$Z = \sum_{i} Z_{i}$. 
Given that one draw of 
$\hat{Z}$ is in $[g_l(n), g_u(n)]$, 
we want to give bounds on another draw.
This is basically a Bayesian problem, 
in that we effectively construct 
a prior over $\EX[\hat{Z}]$.
Two Chernoff bounds suffice 
for the argument. 

\begin{claim}\label{clm:cher} 
  Consider a single draw of 
  $\hat{Z}$ and suppose that 
  $\hat{Z} \in [g_l(n), g_u(n)]$,
  With probability at least 
  $1 - n^{-c'b/5}$, 
  $\EX[\hat{Z}] \in [g_l(n)/6, 2g_u(n)]$.
\end{claim}

\begin{proof} 
Apply Chernoff 
bounds~\cite[Theorem~1.1]{DubhashiPa09}:
if $\mu := \EX[\hat{Z}] < g_l(n)/6$,
then $2e\mu < g_l(n)$, so 
\[
  \Pr[\hat{Z} \geq g_l(n)] < 2^{-g_l(n)} < n^{-c'b/2},
\]
noting that 
$g_l(n) = f_l(n) - b > (c'b/2)\log_2 n$.
With probability at least 
$1 - n^{c'b/2}$, if $\hat{Z} \geq g_l(n)$, 
then $\EX[\hat{Z}] \geq g_1(n)/6$.
We repeat the argument with 
a lower tail Chernoff bound. 
If $\mu > 2g_u(n)$, 
\[
  \Pr[\hat{Z} \leq g_u(n)] \leq \Pr[\hat{Z} \leq (1 - 1/2)\mu] < 
  e^{-g_u(n)/4} < n^{-c'b/4}.
\]
With probability at least 
$1 - n^{c'b/4}$, if $\hat{Z} \leq g_u(n)$, 
then $\EX[\hat{Z}] \leq 2g_u(n)$.
Now take a union bound.
\end{proof}

In Claim~\ref{clm:cher},
we conditioned on a fixed $Q_J$, 
but the bound holds irrespective 
of $Q_J$, and hence is
holds unconditionally. 
Therefore, for a fixed $J$, 
with probability at least 
$1 - n^{-c'b/5}$ over $Q^* \sim \cD$, if 
$\hat{Z} \in [g_l(n), g_u(n)]$, then 
$\EX[\hat{Z}] \in [g_l(n)/6, 2g_u(n)]$.
Given that $|\hat{Z} - Z| \leq b$, this 
implies: 
if ${Z} \in [f_l(n), f_u(n)]$,  then
$\EX[Z] \in [f_l(n)/7, 3f_u(n)]$.

There are $O(n^b)$ choices for $J$, 
so by a union bound 
the above holds for all $J$ simultaneously
with probability at least $1 - n^{-c'b/6}$. 
Suppose we choose a $Q$ with this property,
and consider drawing $P \sim \cD$. 
This is effectively an independent draw of $Z$,
so applying Chernoff bounds again, 
for sufficiently small constants 
$\alpha, \beta$,
\[
  \Pr\Bigl[\sum_{i \leq n} \chi(p^*_i,\reg(Q^*_{J})) 
    \in [\alpha f_l(n), f_u(n)/\alpha]\Bigr] > 
    1 - \exp(-\beta f_l(n)) > 1 - n^{-3}.
\]
\end{proof}

\subsection{Proof of Lemma~\ref{lem:canonicalDir}} 
\label{sec:canonicalDir}

We sample a random input $Q^* \sim \cD$
and take the $(\log^4 n)$-level 
of $Q^*$. Let $H'$
the upper hull of its vertices; 
see Fig.~\ref{fig:arrangement}.
\begin{figure}
  \centering
  \includegraphics{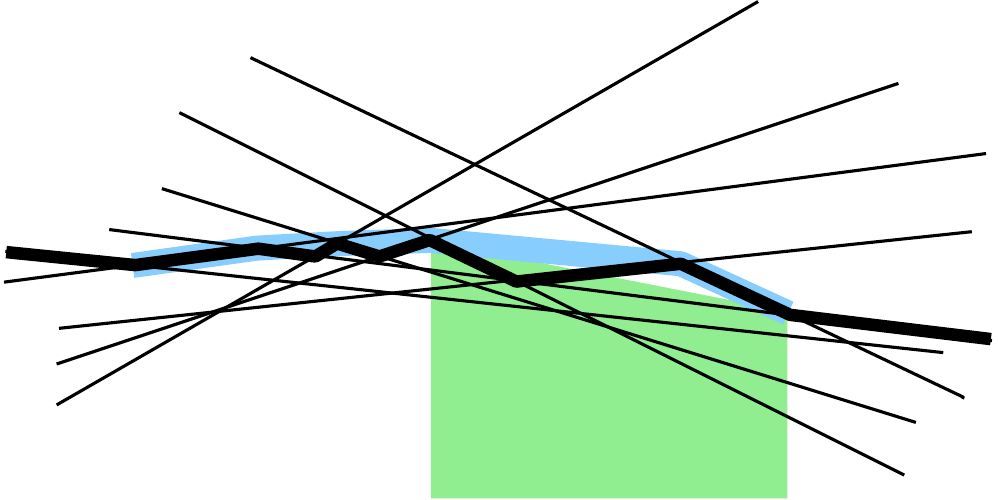}
  \caption{The arrangement of $Q^*$: the 
    dark black line is $\lev_4(Q^*)$. 
    The thick lighter line
    is $H'$, the upper hull of the vertices 
    in the level. The shaded region 
    is a possible trapezoid $\tau_j$.}
  \label{fig:arrangement}
\end{figure}
\begin{claim}\label{clm:H_properties}
  The hull $H'$ has the following properties:
  \begin{asparaenum}
    \item 
      the curve $H'$ lies below $\lev_{2\log^4 n}(Q^*)$;
     \item 
       each line of $Q^*$ either supports 
       an edge of $H'$ or intersects 
       it at most twice; and
     \item 
       $H'$ has $O(n)$ vertices.
  \end{asparaenum}
\end{claim}
\begin{proof} 
Let $p$ be any point 
on $H'$, and let $a, b$
be the closest vertices of $H'$ with
$x(a) \leq x(p) \leq x(b)$.
Any line in $Q^*$ below
$p$ must also be below $a$ or $b$. 
There are exactly $\log^4n$ lines
under $a$ and under $b$, as they lie
on the $(\log^4n)$-level. Hence,
there are at most $2\log^4n$ lines 
below $p$.  The second property 
is a direct consequence of convexity.
The third property follows from the 
second: every vertex of $H'$ lies 
on some line of $Q^*$, and hence
there can be at most $2n$ vertices.
\end{proof}

Let $r_0, \dots, r_k$ be the 
points given by every $\log^2 n$-th 
point in which a line of $Q^*$ meets 
$H'$ (either as an intersection point 
or as the endpoint of a segment), 
ordered from right to left.
By Claim~\ref{clm:H_properties}(2), there 
are $k = O(n/\log^2 n)$ points $r_i$. 
Let $H$ be their upper upper hull.
Clearly, $H$ lies below $H'$.
Draw a vertical downward ray through each vertex
$r_i$. This subdivides the region below 
$H$ into semi-unbounded trapezoids
$\tau_0, \tau_1, \dots$ with
the following properties:
(i) each vertical boundary ray 
  of a trapezoid $\tau_j$ is intersected
  by at least $\log^4 n$ and at most
  $2\log^4 n$ lines of $Q^*$ 
  (Claim~\ref{clm:H_properties}(1)); and 
(ii) the upper boundary
  segment of each $\tau_j$ is intersected 
  by at most $\log^2 n$ lines in 
  $Q^*$ (by construction); see 
  Fig.~\ref{fig:arrangement}.
The next claim follows from an 
application of Lemma~\ref{lem:rand-const}.

\begin{claim}\label{clm:taui}
  With probability at least 
  $1-n^{-4}$ \textup(over 
  $Q$\textup), the 
  following holds for all 
  trapezoids $\tau_j$: 
  generate an
  independent $P^* \sim \cD$.
  \begin{asparaenum}
    \item With probability 
    \textup(over $P^*$\textup) 
    at least $1 - n^{-3}$, 
    there exists a line in $P^*$ 
    that intersects both 
    boundary rays of $\tau_j$; and
    \item with probability 
    \textup(over $P^*$\textup) at 
    least $1 - n^{-3}$, at 
    most $\log^5 n$ lines 
    of $P^*$ intersect $\tau_j$.
\end{asparaenum}
\end{claim}

\begin{proof} 
We apply Lemma~\ref{lem:rand-const} 
for both parts.
For a set $L = \{\ell_1, \ell_2, \ell_3, \ell_4\}$ 
of four lines,
define $\reg(L)$ as the downward 
unbounded vertical trapezoid formed by 
the segment between the intersection
points of $\ell_1, \ell_2$ and 
$\ell_3, \ell_4$. 
All trapezoids $\tau_j$ 
are of this form, with $L$
a set of four lines from $Q^*$.
Set $\chi(\ell, \tau)$ (for line $\ell$ 
and trapezoid $\tau$)
to $1$ if $\ell$ intersects both 
parallel sides of $\tau$, and $0$ otherwise.

Since $\tau_j$ is an unbounded
trapezoid, a line that 
intersects it either 
intersects the upper segment
or intersects both boundary rays. 
In our sample $Q$, the number 
of lines with the former
property is at most $\log^2n$
and the number of lines with 
the latter property is in 
$[\log^4n,4\log^4n]$.
Hence, the sum 
$\sum_{i=1}^n \chi(q^*_i, \tau_j)$ is at least 
$\log^4 n - \log^2 n \geq (1/2)\log^4 n$ 
and at most $5\log^4n$.
By Lemma~\ref{lem:rand-const}, 
the number of lines in $P^*$ intersecting
both vertical lines is 
$\Omega(\log^4n)$ with probability 
at least 
$1 - n^{-3}$.

For the second part, we 
set $\chi(\ell, \tau) = 1$ 
if $\ell$ 
intersects $\tau$, and $0$ otherwise. 
Any line that intersects $\tau_j$
must intersect one of the 
vertical boundaries,
so 
$\sum_{i=1}^n \chi(q^*_i, \tau_j) \in [\log^4n, 4\log^4 n]$.
By Lemma~\ref{lem:rand-const}, 
the number of lines in $P^*$ intersecting
$\tau_j$ is $O(\log^4n)$ 
with probability at least $1 - n^{-3}$.
\end{proof}

Each point $r_i$ is dual 
to a line $r_i^*$. We define the
directions in $\textbf{V}$ by 
taking upward unit normals to 
the lines $r_j^*$.  (Since the $r_i$'s 
are ordered from right to left, this 
gives $\textbf{V}$ in clockwise order.)
These directions can 
be found in $O(n \poly\log(n))$
time: we can compute $\lev_{\log^4 n}(P^*)$ 
and its upper hull $H'$ in $O(n \poly\log n)$ 
time~\cite{ColeShYa87,Dey98}.
To determine the points $r_j$, 
we perform $O(n)$ binary searches over
$H'$, and then sort the intersection points. 
When $r_j$ is
known, $v_j$ can be found 
in constant time.

Now consider a random $P$. 
The $\textbf{V}$-extremal 
vertices $e_i$ and $e_{i+1}$ 
correspond to the lowest line in 
$P^*$ that intersects the left and the
right boundary ray of $\tau_{i}$. 
The number of extremal points between
$e_i$ and $e_{i+1}$ is the number of edges 
on the lower envelope
of $P^*$ between $x(r_{i})$ and 
$x(r_{i+1})$. By Claim~\ref{clm:taui}(1),
this lower envelope lies entirely 
inside $\tau_{i}$ with probability
at least $1-n^{-3}$. 
By Claim~\ref{clm:taui}(2) 
(and a union bound), the number $Y_i$ of
extremal points between $e_i$ and
$e_{i+1}$ is at most $\log^5 n$ 
with probability at least $1-2n^{-3}$.
Thus, 
\begin{align*}
  &\EX[X_i \log (Y_i+1)]\\
  &\leq 
    \EX[X_i \log(\log^5n + 1) \mid Y_i \leq \log^5 n]
    \Pr[Y_i \leq \log^5n] + 
    \EX[X_i \log (Y_i+1) \mid Y_i > \log^5n]
    \Pr[Y_i > \log^5n]\\
  &\leq \EX[X_i  \mid Y_i \leq \log^5n]
    \Pr[Y_i \leq \log^5n]O(\log\log n) + 
    O(n^2) (1/2n^3)\\
    &\leq \EX[X_i] O(\log\log n) + O(1).
\end{align*}
Adding over $i$, 
\begin{multline*}
  \sum_{i=1}^k \EX[X_i \log (Y_i+1)]
  \leq 
  \sum_{i=1}^k \EX[X_i]O(\log\log n) + O(1) \\
  = \EX\Bigl[ \sum_{i=1}^k X_i \Bigr] O(\log\log n) + O(n)
  = O(n\log\log n).
\end{multline*}

\subsection{Proof of Lemma~\ref{lem:ellTail}}\label{sec:ellTail}

To compute the canonical lines 
$\ell_j$ for the directions
$v_j \in \textbf{V}$, we 
consider again the dual 
sample $Q^*$.
Let $s_j$ be the point on
$\lev_{\gamma c\log n}(Q)$ with 
the same $x$-coordinate as 
$r_{j}$, where $\gamma > 0$ 
is a sufficiently
small constant.
Set $\ell_j = s_j^*$.
Then 
$\ell_j$ is normal to $v_j$,
and the construction takes  
$O(n\poly\log n)$ time. 
We restate the
main technical part 
of Lemma~\ref{lem:ellTail}.

\begin{lemma}
  With probability at least $1-n^{-4}$ 
  over the construction, 
  for every $\ell_j$, 
  \[
    \Pr_{P \sim \cD}[|\ell_j^+ \cap P| \in [1, c\log n]] \geq 1 - n^{-3}.
  \]
\end{lemma}

\begin{proof} 
A point $p$ lies in $\ell_j^+$ if
and only if $p^*$
intersects the downward 
vertical ray $R_j$ from $s_j$.
We set up an application of 
Lemma~\ref{lem:rand-const}. 
For a pair of lines $\ell_1, \ell_2$ 
(all in dual space), define 
$\reg(\ell_1,\ell_2)$ as 
the downward vertical ray
from $\ell_1 \cap \ell_2$. 
Every $s_j$ 
is formed by the intersection 
of two lines from $Q^*$. For 
such a region $R_j$
and line $\ell'$, set 
$\chi(\ell',R)$ to be $1$ if 
$\ell'$ intersects $R_j$
and $0$ otherwise.
By construction, 
$\sum_i \chi(q^*_i, R_j) = \gamma c\log n$.
We apply Lemma~\ref{lem:rand-const}. 
With probability at least 
$1-n^{-4}$ over $Q^*$ (for sufficiently 
large $c$ and
small enough $\gamma$),
$\Pr_{P \sim \cD}[\sum_i \chi(p^*_i,R_j) \in [1, c\log n]] \geq 1 - n^{-3}$.
\end{proof}

\section*{Acknowledgements} \label{sec:ack}

C. Seshadhri was supported 
by the Early Career LDRD 
program at Sandia National
Laboratories. Sandia National 
Laboratories is a multi-program laboratory
managed and operated by Sandia 
Corporation, a wholly owned subsidiary of
Lockheed Martin Corporation, for 
the U.S. Department of Energy's National
Nuclear Security Administration 
under contract DE-AC04-94AL85000.
W. Mulzer was supported in 
part by DFG grant MU/3501/1.

\bibliographystyle{abbrv}
\bibliography{si-ch}
\end{document}